\documentclass[a4paper,11pt]{article}

\input{Preamble.tex}

\begin{document}

\title{The Parameterized Complexity of Motion Planning for Snake-Like Robots}

\author{
Siddharth Gupta\thanks{Ben-Gurion University of the Negev, Israel. \texttt{siddhart@post.bgu.ac.il}}
\and Guy Sa'ar\thanks{Ben-Gurion University of the Negev, Israel. \texttt{saag@bgu.ac.il}}
\and Meirav Zehavi\thanks{Ben-Gurion University of the Negev, Israel. \texttt{meiravze@bgu.ac.il}}
}
\maketitle

\begin{abstract}
We study the parameterized complexity of a variant of the classic video game Snake that models real-world problems of motion planning. Given a snake-like robot with an initial position and a final position in an environment (modeled by a graph), our objective is to determine whether the robot can reach the final position from the initial position without intersecting itself. Naturally, this problem models a wide-variety of scenarios, ranging from the transportation of linked wagons towed by a locomotor at an airport or a supermarket to the movement of a group of agents that travel in an ``ant-like'' fashion and the construction of trains in amusement parks. Unfortunately, already on grid graphs, this problem is \textsf{PSPACE}-complete [Biasi and Ophelders, 2016]. Nevertheless, we prove that even on general graphs, the problem is solvable in time $k^{\OO(k)}|I|^{\OO(1)}$ where $k$ is the size of the snake, and $|I|$ is the input size. In particular, this shows that the problem is {\em fixed-parameter tractable (\FPT)}. Towards this, we show how to employ color-coding to sparsify the {\em configuration graph} of the problem to have size $k^{\OO(k)}|I|^{\OO(1)}$ rather than $|I|^{\OO(k)}$. We believe that our approach will find other applications in motion planning. Additionally, we show that the problem is unlikely to admit a polynomial kernel even on grid graphs, but it admits a treewidth-reduction procedure. To the best of our knowledge, the study of the parameterized complexity of motion planning problems (where the intermediate configurations of the motion are of importance) has so far been largely overlooked. Thus, our work is pioneering in this regard.
\end{abstract}

\thispagestyle{empty}
\newpage
\pagestyle{plain}
\setcounter{page}{1}
 
\section{Introduction}\label{sec:intro}

A basic {\em single-agent movement problem} can be modeled by an {\em agent} (representing a robot or a person) that has an initial {\em state} (also called {\em configuration}), a description of valid {\em transitions} between states, and a task to accomplish. Common tasks are to reach some desired position or geographical location while avoiding unwelcome (mobile or static) obstacles, collecting or distributing a set of items, or rearranging the environment to be of a specific form. The agent itself might have various features or restrictions, which are reflected in the definition of states and transitions. Arguably, given that we handle physical objects, a basic requirement is that the agent must never {\em intersect} itself as well as other objects. When several agents are present (in a multi-agent movement problem), the coordination between them also plays a major role~(see \cite{DBLP:conf/compgeom/DemaineFKSM18} and the references within).
Problems based on motion planning are ubiquitous in various aspects of modern life. In recent years, the study of such problems has gained increasing interest from both practical and theoretical points of views~\cite{DBLP:journals/ai/SchwartzS88,GALCERAN20131258,motionSurvey,DBLP:journals/tiv/PadenCYYF16,motionSurvey1}. Unfortunately, the perspective of parameterized complexity---a central paradigm to design algorithms for computationally hard problem---has been largely overlooked in this context. In this paper, we present a comprehensive picture of the parameterized complexity of a single-agent movement problem called the {\sc Snake Game} problem, whose formulation is inspired by a classic video game of the same name.

In the past decade, the study of the theory behind the computational complexity of puzzles (such as video games) has become very popular~\cite{DBLP:books/daglib/0023750,DBLP:journals/icga/KendallPS08,DBLP:journals/mst/Viglietta14}. Such puzzles are often based on motion planning problems that can model tasks to be performed by agents in real-life scenarios. Moreover, their formulations are frequently simple enough to provide a clean abstraction of basic issues in this regard, therefore making them attractive for laying foundations for general analysis. For example, a very long line of works analyzed the complexity of various push-block puzzles (see \cite{DBLP:conf/ciac/DemaineGL17} and the references within), where a box-shaped agent with the ability to push/pull other boxes should utilize its ability in order to reach one position from another. We remark that more often than not, studies of the theory behind the computational complexity of puzzles only assert the \NP-hardness or \textsf{PSPACE}-completeness of the puzzle at hand. 

The classic game Snake is among the most well-known video games that involve the motion of a single agent. The game dates back to 1978, and has enjoyed implementation across a wide range of platforms since then. Unlike most other video games, the popularity of Snake has hardly decreased despite its age---indeed, new versions of Snake still appear to this day. We study the parameterized complexity of a variant of Snake that was introduced by Biasi and Ophelders~\cite{DBLP:journals/tcs/BiasiO18}, which models real-world problems of motion planning for agents of a ``snake-like'' shape. Given a snake-like robot with an initial position and a final position in an environment (modeled by a graph), our objective is to determine whether the robot can reach the final position from the initial position without intersecting itself. Roughly speaking, the position of the robot is modeled by a simple ordered path in the graph, and one position $P$ is reachable (in one step) from another position $P'$ if the path $P$ is obtained from $P'$ by adding one vertex to the beginning of $P'$ and removing one vertex from its end. The (immobile) obstacles in the environment are implicitly encoded in the input graph---specifically, ``obstacle-free'' physical locations are represented by vertices, and edges indicate which locations are adjacent. Note that the graph might not be planar as in real-life scenarios the environment is often not merely a plane.

Nowadays, robots of a ``snake-like'' shape are of substantial interest--in particular, they are built and used in practice for medical operations~\cite{degani2006highly,mitReviewSnake,Berthet-Rayne2018} as well as various inspection and rescue missions on both land and water~\cite{PFOTZER2017123,1302522,LU2016101}. A snake-like shape and serpentine locomotion offer immediate advantages for such purposes; the restricted area of mobility also makes the requirement of the robot to avoid intersecting itself and other obstacles a highly non-trivial issue that is mandatory to take into account.
Moreover, the {\sc Snake Game} problem is a natural abstraction to model a wide-variety of other scenarios, which range from the transportation of linked wagons towed by a locomotor at an airport or a supermarket to the movement of a group of agents that travel in an ``ant-like'' fashion and the construction of trains in amusement parks. 

Biasi and Ophelders~\cite{DBLP:journals/tcs/BiasiO18} proved that the {\sc Snake Game} problem is \textsf{PSPACE}-complete even on {\em grid graphs} (see Section~\ref{sec:prelims}). Additionally, they considered the version aligned with the video game, where ``food'' items are located on vertices. Here, the task is not to reach a pre-specified position, but to collect all food items by visiting their vertices---when a food item is collected, the size of the snake increases by a fixed integer $g\geq 1$. They showed that this version is \NP-hard even on rectangular grid graphs without ``holes'', and \textsf{PSPACE}-complete even when there are only two food items, or the initial size of the snake is $1$.

\subsection{Our Contribution and Methods}

We present a comprehensive picture of the parameterized complexity of the {\sc Snake Game} problem parameterized by the size of the snake, $k$.\footnote{Definitions of basic notions in Parameterized Complexity can be found in Section \ref{sec:prelims}.} Arguably, the choice of this parameter is the most natural and sensible one because in real-life scenarios as those mentioned above, the size of the snake-like robot is likely to be substantially smaller than the size of the environment. To some extent, our paper can be considered as pioneering work in the study of the parameterized complexity of motion planning problems where intermediate configurations are of importance (which has so far been largely neglected, see the next subsection), and may lay the foundations for further research of this topic.

\medskip
\noindent{\bf FPT Algorithm.} Our contribution is threefold. Our main result is the proof that the {\sc Snake Game} problem is {\em fixed-parameter tractable (\FPT)} with respect to $k$. Specifically, we develop an algorithm that solves {\sc Snake Game} in time $k^{\OO(k)}\cdot n^{\OO(1)}$ where $n$ is the number of vertices in the input graph. We remark that our algorithm can also output the length of the shortest ``route'' from the initial position to the final position (if any such route exists) within the same time complexity. The design of our algorithm involves a novel application of the method of color-coding to sparsify the ($(k-1)$-th power of the) {\em configuration graph} of the problem. Roughly speaking, the configuration graph is the directed graph whose vertices represent the positions of the snake in the environment, and where there is an arc from one vertex $u$ to another vertex $v$ if the position represented by $v$ is reachable in one step from the position represented by $u$. Observe that the number of vertices of the configuration graph equals the number of (simple) ordered paths on $k$ vertices in the input graph, which can potentially be {\em huge}---for example, if the input graph is a clique, then there are ${n \choose k}k!$ configurations.

We first present a handy characterization of the reachability of one configuration from another in $t\geq 1$ steps, based on which we elucidate the structure of certain triplets of configurations. Then, we perform several iterations where we color the vertices of the input graph based on the method of color-coding~\cite{DBLP:journals/jacm/AlonYZ95}, but where order between some of the colors is of importance. Within each coloring iteration, we test for every pair of vertices in the input graph whether there exists a particular path on $k$ vertices between them---in which case we pick one such path. The collection of paths found throughout these iterations form the vertex set of our new configuration graph (in addition to the initial and final positions), while our characterization of reachability determines the arcs. In particular, this new configuration graph, unlike the original configuration graph, has only $k^{\OO(k)}\cdot n^{\OO(1)}$ vertices. We prove that this new configuration graph is a valid ``sparsification'' of the $(k-1)$-th power of the original configuration graph---specifically, for any $t\geq 1$, the initial position can reach the final position in the original configuration graph within $(k-1)t$ steps if and only if the initial position can reach the final position in the new configuration graph within $t$ steps. Clearly, this is insufficient because the initial position may reach the final position in the original configuration graph within a number of steps that is not a multiple of $(k-1)$---however, we find that this technicality can be overcome by adding, to the new configuration graph, a small number of new vertices as well as arcs ingoing from these new vertices to the vertex representing the final position.

\medskip
\noindent{\bf Kernelization.} Our second result is the proof that the {\sc Snake Game} problem is unlikely to admit a polynomial kernel even on grid graphs. For this purpose, we present a non-trivial {\em cross-composition} (defined in Section \ref{sec:prelims}) from {\sc Hamiltonian Cycle} on grid graphs to the {\sc Snake Game} problem. Our construction is inspired by the proof of \NP-hardness of the version of {\sc Snake Game} on grid graphs without holes in the presence of food, given by Biasi and Ophelders~\cite{DBLP:journals/tcs/BiasiO18}. Roughly speaking, given $t$ instances of {\sc Hamiltonian Cycle} on grid graphs, our construction of an instance of {\sc Snake Game} is as follows. We position the $t$ input grid graphs so that they are aligned to appear one after the other, and connect them as ``pendants'' from a long path $L$ placed just above them. The initial position of the snake is at the beginning of $L$, and its final position is at a short path that ``protrudes'' from the beginning of $L$. Further, the size of the snake is set to be $n$, the number of vertices in each of the input grid graphs. Intuitively, we show that the snake can reach the final position from the initial position if and only if the snake can enter and exit one of the input grid graphs. In particular, to reach the final position, the snake has to find a place where it can ``turn around'', which can only be (potentially) done inside one of the input grid graphs. Specifically, entering and exiting one of the input grid graphs, say, $G_i$, would imply that at some point of time, the snake must be fully inside $G_i$ and hence exhibit a Hamiltonian cycle within it.

\medskip
\noindent{\bf Treewidth-Reduction.} Our last result is a treewidth-reduction procedure for the {\sc Snake Game} problem. More precisely, we develop a polynomial-time algorithm that given an instance of {\sc Snake Game}, outputs an equivalent instance of {\sc Snake Game} where the treewidth of the graph is bounded by a polynomial in $k$. Our procedure is based on the irrelevant vertex technique~\cite{DBLP:journals/jct/RobertsonS95b}. First, we exploit the relatively recent breakthrough result by Chekuri and Chuzhoy~\cite{DBLP:journals/jacm/ChekuriC16} that states that, for any positive integer $t\in\mathbb{N}$, any graph whose treewidth is at least $d\cdot t^c$ (for some fixed constants $c$ and $d$)\footnote{Currently, the best known bound on $c$ is $10$, given by Chuzhoy in \cite{DBLP:conf/soda/ChuzhoyT19}.} has a $t\times t$-grid as a minor, and hence also a so called $t$-wall as a subgraph (see Section \ref{sec:treewidth}). We utilize this result to argue that if the treewidth of our input graph is too large, then it has a $ck$-wall as a subgraph (for some fixed constant $c$) such that no vertex of this $ck$-wall belongs to the initial or final positions of the snake. The main part of our proof is a re-routing argument that shows that in such a wall, we can arbitrarily choose any pair of adjacent vertices, contract the edge between them and thereby obtain an equivalent instance of the {\sc Snake Game} problem. Thus, as long as we do not yet have a graph of small treewidth at hand, we can efficiently find an edge to contract, and eventually obtain a graph of small treewidth.

\subsection{Related Works in Parameterized Complexity}

To the best of our knowledge, close to nothing is known on the parameterized complexity of problems of planning motion of agents where the motion plays an actual role. By this comment, we mean that the problem does not only have an initial state and a final state, and intermediate states (with transitions between them) are not defined/immaterial. When only an initial state and a final state are present, the problem is actually a {\em static} problem and the term movement refers to distances rather than motion. For example, a static version in this spirit of the {\sc Snake Game} problem (that neglects intermediate states and hence does not enforce the requirement that the snake should never intersect itself) would just ask whether the vertex on which lies the tail of the snake in the final position is reachable from the vertex on which lies the head of the snake in the initial position, which can be directly solved in linear time using BFS. We also remark that almost all problems studied from the viewpoint of Parameterized Complexity are \NP-complete rather than \textsf{PSPACE}-hard (with several notable exceptions such as~\cite{DBLP:conf/lics/PanV06,DBLP:journals/jcss/HaanS17}).

While there is a huge line of works that explore the classical complexity and approximability of problems of planning motion of agents (see the beginning of the introduction), the only work in Parameterized Complexity that we are aware of is by Cesati and Wereham \cite{conf/IEEE/Cesati995}. Roughly speaking, the input of the problem in \cite{conf/IEEE/Cesati995} consists of a robot and obstacles that are each represented by a set of polyhedrons. Specifically, the input includes a set $O$ of polyhedrons (the obstacles), as well as a set $P$ of polyhedrons (the robot) that are freely linked together at a set of linkage vertices $V$ such that $P$ has $k$ degrees of freedom of movement. The objective is to decide whether a given final position of the robot is reachable from a given initial position of the robot where the robot is allowed to intersect neither itself nor the obstacles. Reif~\cite{DBLP:conf/focs/Reif79} proved that this problem (in three-dimensional Euclidean space) is \textsf{PSPACE}-hard, and Cesati and Wereham \cite{conf/IEEE/Cesati995} adapted Reif's proof and showed that the problem is \textsf{W[SAT]}-hard when parameterized by $k$ and hence unlikely to be \FPT.

Lastly, when intermediate states are immaterial, it is less clear which problem is considered to be centered around motion rather than euclidean/graph distances and static constraints---for example, in such scenarios finding a shortest path or a path with specific restrictions in a given graph may be highly relevant to motion planning (see e.g.~\cite{DBLP:conf/icalp/EibenK18}), but is not centered around the motion itself in the sense above. In this context, it is noteworthy to mention the comprehensive work in Parameterized Complexity by Demaine, Hajiaghayi and Marx \cite{DBLP:journals/talg/DemaineHM14}. Roughly speaking, the authors addressed multi-agent problems on graphs where the agents have types, there is an initial position (being a single vertex) for each agent, and a desired configuration that all agents should form (e.g., every agent of type ``client'' should have an agent of type ``facility'' nearby). Under some restrictions (e.g., some agents may move only a short distance), the objective is to make the minimum amount of movement so that afterwards the agents will form the desired configuration. Clearly, movement is a central component of this problem, but in a static sense---the intermediate configurations of the agents are immaterial (while moving to form the desired configuration, the agents can be in any configuration that they wish to, such as being placed all together on the same vertex, or being as far apart as possible). Demaine, Hajiaghayi and Marx \cite{DBLP:journals/talg/DemaineHM14} presented, among other results, a dichotomy that concerns the fixed-parameter tractability of problems of this form parameterized by the number of agents (more precisely, of agents of a certain type) and another parameter related to the desired configuration.
\section{Preliminaries}\label{sec:prelims}
Let $\mathbb{N}_0=\mathbb{N}\cup \{ 0 \}$.
We denote the set $\mathbb{N}_0\times \mathbb{N}_0\times \mathbb{N}_0\times\mathbb{N}_0$ by $\mathbb{N}_0^4$. For any $t \in \mathbb{N}$, we denote the set $\{1,2, \ldots, t\}$ by $[t]$.

\paragraph{Graphs.} Given a graph $G$, let $V(G)$ and $E(G)$ denote its vertex set and edge set, respectively. For a path $P$, let the {\em size} and the {\em length} of $P$ denote the number of vertices and edges in $P$, respectively. 

Let $G$ be an undirected graph. For a vertex $v\in V(G)$, we denote the set of all the vertices adjacent to $v$ in $G$ by $N_G(v)$, i.e. $N_G(v)=\{u\in V(G)~|~\{u,v\}\in E(G)\}$. The {\em edge contraction operation} for an edge $\{ u,v\}\in E(G)$, is the addition of a new vertex $w$ such that $N_G(w) = N_G(u) \cup N_G(v)$ and the deletion of the edge $\{ u,v\}$ and the vertices $u$ and $v$ from $G$. For an edge $e\in E(G)$, we denote the graph obtained from $G$ by contracting $e$ by $G/e$. We denote the new vertex that was created by the contraction by $v_e$. The {\em edge subdivision operation} for an edge $\{ u,v\}\in E(G)$, is the deletion of the edge $\{ u,v\}$ from $G$ and the addition of two new edges $\{ u,w\}$ and $\{ w,v\}$ for a new vertex $w$. For an undirected graph $H$, we say that $G$ {\em can be obtained from $H$ by subdividing edges}, if $G$ can be derived from $H$ by a sequence of edge subdivision operations. A {\em Hamiltonian cycle} in $G$ is a cycle in $G$ that visits every vertex of $G$ exactly once. For other standard notations not explicitly defined here, we refer to the book \cite{Diestelbook}.

For our negative result on kernelization, we make use of the {\sc Hamiltonian Cycle} problem on grid graphs, defined as follows.

\begin{definition}[{\bf Grid Graph}]\label{def:Grid graph}
A {\em grid graph} is a finite undirected graph $G$ with $V(G)\subseteq \{ (i,j) | i,j\in \mathbb{N}_0 \}$ and $\{(i,j),(i',j')\}\in E(G)$ if and only if $|i-i'|+|j-j'|=1$.
\end{definition}

Given $k, \ell, r \in \mathbb{N}$, a {\em $(k \times \ell)$-$r$-grid} is a grid graph $G$ with $V(G) = \{(x,y) | r \leq x\leq r+k-1, r \leq y\leq r+\ell-1 \}$. Given a $(k \times \ell)$-$r$-grid $G$, the four {\em corner vertices} of $G$ are $(r,r), (r, r+\ell-1), (r+k-1, r)$ and $(r+k-1, r+\ell-1)$.
For any $k,\ell \in \mathbb{N}$ we denote $(k \times \ell)$-$1$-grid by $k\times \ell$-grid. 

Given a graph $G$, the objective of the {\sc Hamiltonian Cycle} problem is to decide whether there exists a Hamiltonian cycle in $G$. The {\sc Hamiltonian Cycle} problem is known to be \NP-complete even when restricted to grid graphs~\cite{DBLP:journals/jal/PapadimitriouV84}.

For our result on treewidth, we make use of the {\em treewidth} of a graph $G$, defined as follows.

\begin{definition}[{\bf Treewidth}]
A \emph{tree decomposition} of a graph $G$ is a tree $T$ whose nodes, called \emph{bags}, are labeled by subsets of vertices of $G$. For each vertex $v$ the bags containing $v$ must form a nonempty contiguous subtree of $T$, and for each edge $uv$ at least one bag must contain both $u$ and $v$. The \emph{width} of the decomposition is one less than the maximum cardinality of any bag, and the \emph{treewidth} $\tw(G)$ of $G$ is the minimum width of any of its tree decompositions. 
\end{definition}

\paragraph{Snake Game.} Towards the definition of {\sc Snake Game}, we begin by defining the notion of a configuration in this context.

\begin{definition}[{\bf $(G,k)$-Configuration}]\label{def:config}
Let $G$ be an undirected graph, and let $k\in \mathbb{N}$. A {\em $(G,k)$-configuration} is a tuple $(v_1,v_2,\ldots,v_k)$ where $v_i\in V(G)$ for every $1\leq i\leq k$, which satisfies the following conditions.
\begin{itemize}
\item For every $1\leq i\leq k-1$, we have that $\{v_i,v_{i+1}\}\in E(G)$.
\item For every $1\leq i<j\leq k$, we have that $v_i\neq v_j$.
\end{itemize}
\end{definition}

Given a configuration $\mathsf{conf}$, let $V(\mathsf{conf})$ denotes the set of its vertices. Intuitively, a $(G,k)$-configuration $(v_1,v_2,\ldots,v_k)$ is the sequence of vertices of a simple path on $k$ vertices in $G$ when traversed from one endpoint to the other; the path is termed a {\em snake}, and the vertices $v_1$ and $v_k$ are termed the {\em head} and {\em tail} of the snake, respectively. For the sake of brevity, whenever $G$ and $k$ are clear from the context, we refer to a $(G,k)$-configuration simply as a {\em configuration}. We proceed to define how a snake is permitted to ``move'' from one position to another---that is, we define a {\em transition} (in one step) from one configuration to another. 

\begin{definition}[{\bf $1$-Transition}]\label{def:MoveToConfInOne}
Let $G$ be an undirected graph, and let $k\in \mathbb{N}$. Let $\mathsf{conf}=(v_1,v_2,\ldots,v_k)$ and $\mathsf{conf'}=(v_1',v_2',\ldots,v_k')$ be any two configurations. We say that the pair {\em $(\mathsf{conf},\mathsf{conf'})$} is a {\em $1$-transition} if and only if the following conditions are satisfied.
\begin{itemize}
\item For every $1\leq i\leq k-1$, $v_1'\neq v_i$. 
\item For every $2\leq i\leq k$, $v_i'=v_{i-1}$.
\end{itemize} 
\end{definition}

We naturally extend the definition of a transition in one step to the definition of a transition in $\ell$ steps as follows.

\begin{definition}[{\bf $\ell$-Transition}]\label{def:MoveToConfInEll} 
Let $G$ be an undirected graph, and let $k, \ell \in \mathbb{N}$. Let $\mathsf{conf}=(v_1,v_2,\ldots,v_k)$ and $\mathsf{conf'}=(v_1',v_2',\ldots,v_k')$ be any two configurations. We say that the pair {\em $(\mathsf{conf},\mathsf{conf'})$} is an {\em $\ell$-transition} if there exists a tuple $(\mathsf{conf}_1=\mathsf{conf},\mathsf{conf}_2,\ldots,\mathsf{conf}_{\ell+1}=\mathsf{conf'})$ of $\ell+1$ configurations such that, for every $1\leq i\leq \ell$, the pair $(\mathsf{conf}_i,\mathsf{conf}_{i+1})$ is a $1$-transition. In that case, we also say that the tuple $(\mathsf{conf}_1=\mathsf{conf},\mathsf{conf}_2,\ldots,\mathsf{conf}_{\ell+1}=\mathsf{conf'})$ is an {\em $\ell$-transition}.  
\end{definition}

Before we define the {\sc Snake Game} problem formally, we need one more definition concerning the reachability of one configuration from another, based on the definition of a transition.

\begin{definition}[{\bf Reachability}]\label{def:MoveToconf} 
Let $G$ be an undirected graph, and let $k\in \mathbb{N}$. Let $\mathsf{conf}$ and $\mathsf{conf'}$ be any two configurations. 
We say that $\mathsf{conf}$ {\em can reach} $\mathsf{conf'}$ (alternatively, $\mathsf{conf'}$ {\em is reachable from} $\mathsf{conf}$) if the pair {\em $(\mathsf{conf},\mathsf{conf'})$} is an $\ell$-transition for some $\ell\in\mathbb{N}$.
\end{definition}

We are now ready to give the formal definition of the {\sc Snake Game} problem.

\begin{definition}[{\bf Snake Game}]\label{def:SnakeGame}
An instance of {\sc Snake Game} is quadruple $\mathbf{SG}=\langle G,k,\mathsf{init},\mathsf{fin}\rangle$ where $G$ is an undirected graph, $k\in \mathbb{N}$, and $\mathsf{init}$ and $\mathsf{fin}$ are two configurations. We say that $\mathbf{SG}$ is a \yes-instance if $\mathsf{init}$ can reach $\mathsf{fin}$; otherwise, we say that $\mathbf{SG}$ is a \no-instance. By {\em solving} an instance of {\sc Snake Game}, we mean that we correctly determine whether it is a \yes-instance or a \no-instance.
\end{definition}

We now define a new auxiliary graph, the {\em $\ell$-configuration graph}, which will be helpful throughout the paper.

\begin{definition}[{\bf $\ell$-Configuration Graph}]\label{def:confGraph}
Let $G$ be an undirected graph, and let $k, \ell \in \mathbb{N}$. The $\ell$-configuration graph is a directed graph with a vertex for every $(G,k)$-configuration and a directed edge from a vertex $\mathsf{conf}$ to another vertex $\mathsf{conf'}$ if the pair {\em $(\mathsf{conf},\mathsf{conf'})$} is an $\ell$-transition.
\end{definition}

We can also define reachability and feasibility of a ``solution'' (i.e. witness of a \yes-instance) for an instance of {\sc Snake Game} via the $1$-configuration graph by the following simple observation. 

\begin{observation}\label{obs:confGraph}
Let $G$ be an undirected graph, and let $k\in \mathbb{N}$. Let $\mathsf{conf}$ and $\mathsf{conf'}$ be any two configurations. Then $\mathsf{conf}$ {\em can reach} $\mathsf{conf'}$ if and only if there exists a path from $\mathsf{conf}$ to $\mathsf{conf'}$ in the $1$-configuration graph. In particular, an instance $\mathbf{SG}=\langle G,k,\mathsf{init},\mathsf{fin}\rangle$ of {\sc Snake Game} is a \yes-instance if and only if there exists a path from $\mathsf{init}$ to $\mathsf{fin}$ in the $1$-configuration graph. 
\end{observation}

\paragraph{Parameterized Complexity.} A problem $\Pi$ is a {\em parameterized} problem if each problem instance of $\Pi$ is associated with a {\em parameter} $k$. For simplicity, we denote a problem instance of a parameterized problem $\Pi$ as a pair $(I,k)$ where the second argument is the parameter $k$ associated with $I$. The main objective of the framework of Parameterized Complexity is to confine the combinatorial explosion in the running time of an algorithm for an \NP-hard parameterized problem $\Pi$ to depend only on $k$. Formally, we say that $\Pi$ is {\em fixed-parameter tractable (\FPT)} if any instance $(I, k)$ of $\Pi$ is solvable in time $f(k)\cdot |I|^{\OO(1)}$, where $f$ is an arbitrary computable function of $k$. We remark that Parameterized Complexity also provides methods to show that a parameterized problem is unlikely to be \FPT. The main technique is the one of parameterized reductions analogous to those employed in Classical Complexity. Here, the concept of \WO-hardness replaces the one of \NP-hardness.

A companion notion to that of fixed-parameter tractability is the one of a polynomial kernel. Formally, a parameterized problem $\Pi$ is said to admit a {\em polynomial compression} if there exists a (not necessarily parameterized) problem $\Pi'$ and a polynomial-time algorithm that given an instance $(I,k)$ of $\Pi$, outputs an equivalent instance $I'$ of $\Pi'$ (that is, $(I,k)$ is a \yes-instance of $\Pi$ if and only if $I'$ is a \yes-instance of $\Pi'$) such that $|I'|\leq p(k)$ where $p$ is some polynomial that depends only on $k$. In case $\Pi'=\Pi$, we further say that $\Pi$ admits a {\em polynomial kernel}.
For more information on Parameterized Complexity, we refer the reader to book such as \cite{DBLP:series/txcs/DowneyF13,DBLP:books/sp/CyganFKLMPPS15,fomin2018kernelization}.

\paragraph{Non-Existence of a Polynomial Compression.} Our proof of the ``unlikely existence'' of a polynomial compression (and hence also of a polynomial kernel) for {\sc Snake Game}, even when restricted to grid graphs, relies on the well-known notion of cross-composition. We present this notion in a form sufficient for our proof.

\begin{definition} [{\bf Cross-Composition}]\label{def:cross-comp} A (not parameterized) problem $\Pi$ {\em cross-composes} into a parameterized problem $\Pi'$ if there exists a polynomial-time algorithm, called a {\em cross-composition}, that given instances $I_1,I_2,\ldots,I_t$ of $\Pi$ for some $t \in \mathbb{N}$ that are of the same size $s$ for some $s\in\mathbb{N}$, outputs an instance $(I,k)$ of $\Pi'$ such that the following conditions are satisfied.
\begin{itemize}
\item $k\leq p(s)$ for some polynomial $p$ in $s$.
\item $(I,k)$ is a \yes-instance of $\Pi'$ if and only if at least one of the instances $I_1,I_2,\ldots,I_t$ is a \yes-instance of $\Pi$.
\end{itemize}
\end{definition}

\begin{proposition}[\cite{DBLP:journals/jcss/BodlaenderDFH09,DBLP:journals/siamdm/BodlaenderJK14}]\label{prop:noKern}
Let $\Pi$ be an \NP-hard (not parameterized) problem that cross-composes into a parameterized problem $\Pi'$. Then, $\Pi'$ does not admit a polynomial compression, unless \NP$\subseteq $\coNPpoly. 
\end{proposition}
\section{FPT Algorithm on General Graphs}\label{sec:fpt}

In this section, we describe an \FPT\ algorithm for the {\sc Snake Game} problem, which exploits the notion of a $(k-1)$-configuration graph and the method of color-coding in order to sparsify it. We begin by giving two conditions that identify when a pair of configurations is an $\ell$-transition. These conditions will be useful throughout the section, and are stated in the following lemma. In this context, it may be helpful to refer to Figure~\ref{fig:lTransition}, which shows two configurations $\mathsf{conf_s}$ and $\mathsf{conf_e}$ of length $6$ where $(\mathsf{conf_s}, \mathsf{conf_e})$ is a $3$-transition.

\begin{lemma}\label{lem:confGraph}
Let $G$ be an undirected graph, and let $k, \ell \in \mathbb{N}$ such that $\ell \leq k$. Let $\mathsf{conf}=(v_1,v_2,\ldots,v_k)$ and $\mathsf{conf}'=(v_1',v_2',\ldots,v_k')$ be any two configurations. Then, the pair {\em $(\mathsf{conf},\mathsf{conf}')$} is an {\em $\ell$-transition} if and only if the following conditions are satisfied.
\begin{itemize}
	\item For every $1 \leq i \leq \ell$, $v_i' \notin \{v_1,v_2,\ldots,v_{(k+i)-(\ell+1)}\}$.
	\item For every $\ell+1 \leq i \leq k$, $v_i' = v_{i-\ell}$.
\end{itemize}
\end{lemma}

\begin{figure}
	\includegraphics[width=\textwidth]{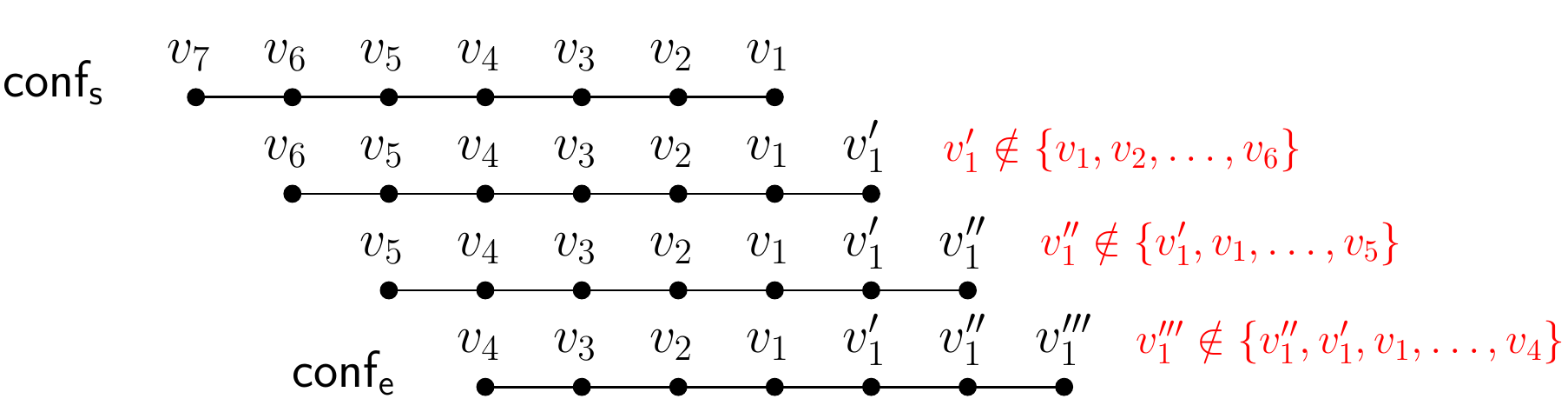}
	\caption{An example of a $3$-transition $(\mathsf{conf_s}, \mathsf{conf_e})$ where $k = 7$.}
	\label{fig:lTransition}
\end{figure}

\begin{proof}
We prove the statement by induction on $\ell$.

\smallskip
\noindent{\bf Base case $(\ell=1)$.} When $\ell=1$, the conditions are as follows:
\begin{itemize}
	\item $v_1' \notin \{v_1,v_2,\ldots,v_{k-1}\}$.
	\item For every $2 \leq i \leq k$, $v_i' = v_{i-1}$.
\end{itemize}
By the definition of a $1$-transition, two configurations form a $1$-transition if and only if the above conditions are true, so the lemma is true for $\ell=1$.

\smallskip
\noindent{\bf Inductive hypothesis.} Suppose that the lemma is true for $\ell=t-1 \geq 1$.

\smallskip
\noindent{\bf Inductive step.} We need to prove that the lemma is true for $\ell=t \geq 2$.\\ 
$(\Rightarrow)$ Let $\mathsf{conf}=(v_1,v_2,\ldots,v_k)$ and $\mathsf{conf}'=(v_1',v_2',\ldots,v_k')$ be two configurations such that the pair {\em $(\mathsf{conf},\mathsf{conf}')$} is a $t$-transition. Then, by Definition~\ref{def:MoveToConfInEll}, there exists a tuple of configurations $(\mathsf{conf}_1=\mathsf{conf},\mathsf{conf}_2,\ldots,\mathsf{conf}_t,\mathsf{conf}_{t+1}=\mathsf{conf}')$ such that for every $1\leq i\leq t$, the pair $(\mathsf{conf}_i,\mathsf{conf}_{i+1})$ is a $1$-transition. 
It follows that $(\mathsf{conf},\mathsf{conf}_t)$ is a $(t-1)$-transition. Let $\mathsf{conf}_t=(w_1,w_2,\ldots,w_k)$. Then, by inductive hypothesis, the following conditions are true:
\begin{itemize}
	\item For every $1 \leq i \leq t-1$, $w_i \notin \{v_1,v_2,\ldots,v_{(k+i)-t}\}$.
	\item For every $t \leq i \leq k$, $w_i = v_{i-t+1}$.
\end{itemize}
Also, $(\mathsf{conf}_t,\mathsf{conf}')$ is a $1$-transition so by Definition~\ref{def:MoveToConfInOne}, the following conditions are true:
\begin{itemize}
	\item $v_1' \notin \{w_1,w_2,\ldots,w_{k-1}\}$.
	\item For every $2 \leq i \leq k$, $v_i' = w_{i-1}$.
\end{itemize}
Combining the above conditions for $(\mathsf{conf},\mathsf{conf}_t)$ and $(\mathsf{conf}_t,\mathsf{conf}')$, we get:
\begin{itemize}
	\item For every $1 \leq i \leq t$, $v_i' \notin \{v_1,v_2,\ldots,v_{(k+i)-(t+1)}\}$.
	\item For every $t+1 \leq i \leq k$, $v_i' = v_{i-t}$.
\end{itemize}

\noindent$(\Leftarrow)$ Let $\mathsf{conf}=(v_1,v_2,\ldots,v_k)$ and $\mathsf{conf}'=(v_1',v_2',\ldots,v_k')$ be two configurations such that the following conditions are satisfied: 
\begin{itemize}
	\item For every $1 \leq i \leq t$, $v_i' \notin \{v_1,v_2,\ldots,v_{(k+i)-(t+1)}\}$.
	\item For every $t+1 \leq i \leq k$, $v_i' = v_{i-t}$.
\end{itemize}
Let $\mathsf{conf}''=(w_1,w_2,\ldots,w_k)$ be a configuration defined as follows:
\begin{itemize}
	\item For every $1 \leq i \leq k-1$, $w_i = v_{i+1}'$.
	\item $w_k = v_{k-t+1}$.
\end{itemize}
From Definition~\ref{def:config}, we know that for every $2 \leq i \leq k$, $v_1' \neq v_i'$. So, from the definition of $\mathsf{conf}''$, we get the following conditions for the pair $(\mathsf{conf}''$, $\mathsf{conf}')$:
\begin{itemize}
	\item $v_1' \notin \{w_1,w_2,\ldots,w_{k-1}\}$.
	\item For every $2\leq i\leq k$, $v_i' = w_{i-1}$.
\end{itemize}
By Definition~\ref{def:MoveToConfInOne}, $(\mathsf{conf}'', \mathsf{conf}')$ is a $1$-transition. Also, by the conditions for $(\mathsf{conf},\mathsf{conf}')$ and the definition of $\mathsf{conf}''$, we get the following conditions for the pair $(\mathsf{conf}$, $\mathsf{conf}'')$:
\begin{itemize}
	\item For every $1 \leq i \leq t-1$, $w_i \notin \{v_1,v_2,\ldots,v_{(k+i)-t}\}$.
	\item For every $t \leq i \leq k$, $w_i = v_{i-t+1}$.
\end{itemize}
By inductive hypothesis, $(\mathsf{conf}, \mathsf{conf}'')$ is a $(t-1)$-transition. Moreover, we already know that $(\mathsf{conf}'',\mathsf{conf}')$ is a $1$-transition. So, by Definition~\ref{def:MoveToConfInEll}, we get that $(\mathsf{conf}, \mathsf{conf}')$ is a $t$-transition.
\end{proof}

The following corollary about $(k-1)$-transitions follows directly from the above lemma.

\begin{corollary}\label{cor:(k-1)Transition}
Let $G$ be an undirected graph, and $k\in \mathbb{N}$. Let $\mathsf{conf}=(v_1,v_2,\ldots,v_k)$ and $\mathsf{conf}'=(v_1',v_2',\ldots,v_k')$ be any two configurations. Then, the pair {\em $(\mathsf{conf},\mathsf{conf}')$} is a {\em $(k-1)$-transition} if and only if the following conditions are satisfied.
\begin{itemize}
	\item For every $1 \leq i \leq k-1$, $v_i' \notin \{v_1,v_2,\ldots,v_i\}$.
	\item $v_k' = v_1$.
\end{itemize}
\end{corollary}

In turn, the following observation about testing whether a pair of configurations is a $(k-1)$-transition follows directly from the above corollary.

\begin{observation}\label{obs:testTransition}
Let $G$ be an undirected graph and $k\in \mathbb{N}$. Let $\mathsf{conf}$ and $\mathsf{conf}'$ be any two configurations. Then, we can test whether $(\mathsf{conf}, \mathsf{conf}')$ is a $(k-1)$-transition in time $\OO(k^2)$ by checking the conditions given in Corollary~\ref{cor:(k-1)Transition} for every vertex of $\mathsf{conf}$ and $\mathsf{conf}'$.
\end{observation} 

Observe that the size of the $(k-1)$-configuration graph is potentially huge. We overcome this difficulty by utilizing the technique of color-coding~\cite{DBLP:journals/jacm/AlonYZ95} to sparsify the configuration graph. We need the following definitions to describe the technique.

\begin{definition}[{\bf Splitter}]\label{def:Splitter}
Let $n, m, t \in \mathbb{N}$ such that $m \leq t$. An {\em $(n,m,t)$-splitter} $\mathcal{F}$ is a family of functions from $[n]$ to $[t]$ such that for every set $W \subseteq [n]$ of size at most $m$, there exists a function $f \in \mathcal{F}$ that is injective on $W$.
\end{definition}

\begin{definition}[{\bf Permuter}]\label{def:Permuter}
Let $n, t \in \mathbb{N}$ and $S$ be a set of size $n$. An {\em $(S,t)$-permuter} $\mathcal{G}$ is a family of functions from $S$ to $[t]$ such that for every ordered set $W \subseteq S$ of size at most $t$ and for any $i, j \in \mathbb{N}$ such that $i \leq j \leq t$ and $j-i+1 = |W|$, there exists a function $g \in \mathcal{G}$ which maps $W$ to the ordered set $\{i,i+1,\ldots,j\}$.
\end{definition}

An efficient construction for an $(S,t)$-permuter easily follows from the well-known efficient construction of splitters~\cite{NaorSS95} as follows.

\begin{lemma}\label{lem:sizePermuter}
For any $n,t \in \mathbb{N}$ and a set $S$ of size $n$, an $(S,t)$-permuter of size $t^{\OO(t)} \log n$ can be constructed in time $t^{\OO(t)} n\log n$.
\end{lemma} 

\begin{proof}
We can construct an $(S,t)$-permuter of size $t^{\OO(t)} \log n$ as follows. Let $\mathcal{F}_1$ be an $(n,t,t)$-splitter. Let $\mathcal{F}_2$ be the family of all permutation functions from $[t]$ to $[t]$. Now consider the family of functions $\mathcal{F}^*$ obtained by composition of $\mathcal{F}_1$ with $\mathcal{F}_2$, i.e. $\mathcal{F}^* = \{f_2 \circ f_1 : f_1 \in \mathcal{F}_1, f_2 \in \mathcal{F}_2\}$. Here, $f_2 \circ f_1$ is the composition of the two functions obtained by performing $f_1$ first and then $f_2$, i.e. $(f_2 \circ f_1)(x) = f_2(f_1(x))$. Let $S = \{s_1, s_2, \ldots, s_n\}$ and $\mathcal{F}^* = \{f_1^*, f_2^*, \ldots, f_{|\mathcal{F}^*|}^*\}$. Now consider the family of functions $\mathcal{G} = \{g_1, g_2, \ldots, g_{|\mathcal{F}^*|}\}$ such that for every $1 \leq j \leq |\mathcal{F}^*|$ and $1 \leq i \leq n$, $g_j(s_i) = f_j^*(i)$. It is easy to see that $\mathcal{G}$ is our required $(S,t)$-permuter.
We now argue about the size of $\mathcal{G}$ and the time taken by the construction. In~\cite{NaorSS95}, Naor {\em et~al.} gave a construction of $\mathcal{F}_1$ of size $e^t t^{\OO(\log t)} \log n$ in time $e^t t^{\OO(\log t)} n\log n$, and the size of $\mathcal{F}_2$ is exactly $t!$. Thus the size of $\mathcal{G}$ is $(e^t t^{\OO(\log t)} \log n) \cdot t! = t^{\OO(t)} \log n$. As it is possible to list the image of a function in linear time in the size of the domain, we can compute $\mathcal{G}$ in time $t^{\OO(t)} n\log n$.
\end{proof}

Given three configurations $\mathsf{conf}$, $\mathsf{conf}'$ and $\mathsf{conf}''$, we now define a special ordering, called {\em triplet order}, on the vertices of the triplet $(\mathsf{conf}$, $\mathsf{conf}'$, $\mathsf{conf}'')$. Note that an example is given below the definition.

\begin{definition}[{\bf Triplet Order}]\label{def:TripletOrder}
Let $G$ be an undirected graph, and let $k \in \mathbb{N}$. Let $\mathsf{conf}=(v_1,v_2,\ldots,v_k)$, $\mathsf{conf}'=(v_1',v_2',\ldots,v_k')$ and $\mathsf{conf}''=(v_1'',v_2'',\ldots,v_k'')$ be any three configurations. Let $S=\{v_k,v_{k-1},\ldots,v_1, v_k',v_{k-1}',\ldots,v_1', v_k'',v_{k-1}'',\ldots,v_1''\}$ be the ordered multi-set defined on the vertices of triplet $(\mathsf{conf},$ $\mathsf{conf}',\mathsf{conf}'')$. Then, the {\em triplet order} of $(\mathsf{conf}, \mathsf{conf}', \mathsf{conf}'')$ is the ordered set obtained from $S$ by first removing the vertices of $\mathsf{conf}$ and $\mathsf{conf}''$ that are common to $\mathsf{conf}'$, and then removing the vertices of $\mathsf{conf}$ that are common to $\mathsf{conf}''$. 
\end{definition}

For example, let $\mathsf{conf}=(v_1, v_2, v_3, v_4, v_5, v_6, v_7)$, $\mathsf{conf}'=(v_1', v_2' = v_6, v_3' = v_7, v_4', v_5' = v_1, v_6' = v_2, v_7' = v_3)$ and $\mathsf{conf}''=(v_1'', v_2'', v_3'' = v_4, v_4'', v_5'', v_6'', v_7'' = v_1')$ be three configurations of length $6$. Then, the triplet order $W$ of $(\mathsf{conf}, \mathsf{conf}',\mathsf{conf}'')$ is $\{v_5, v_7', v_6',v_5',v_4',v_3',v_2', v_1', v_6'', v_5'',v_4'',v_3'',$ $v_2'',v_1''\}$ (see Figure~\ref{fig:lemmaOrdered}). Note that $W$ is an ordered set that contains every distinct vertex of $S$ exactly once.

The following lemma about triplet order will help us to reduce the size of the $(k-1)$-configuration graph. 

\begin{lemma}\label{lem:orderedColoring}
Let $G$ be an undirected graph, and $k, \ell, r \in \mathbb{N}$ such that $\ell, r \leq k-1$. Let $\mathsf{conf}=(v_1,v_2,\ldots,v_k)$, $\mathsf{conf}'=(v_1',v_2',\ldots,v_k')$ and $\mathsf{conf}''=(v_1'',v_2'',\ldots,v_k'')$ be three configurations such that $(\mathsf{conf},\mathsf{conf}')$ is an $\ell$-transition and $(\mathsf{conf}',\mathsf{conf}'')$ is an $r$-transition. Let $W$ be the triplet order of $(\mathsf{conf},\mathsf{conf}',\mathsf{conf}'')$ and $\mathcal{G}$ be a $(V(G), 3k-2)$-permuter. Then, there exists a function $f \in \mathcal{G}$ that satisfies the following conditions (see Figure~\ref{fig:lemmaOrdered}): 
\begin{enumerate}[$(i)$]
	\item\label{prop:oC1} $W$ is mapped to an ordered set $\{i,i+1,\ldots,k,k+1,\ldots,2k-1,\ldots,j\}$ such that $1 \leq i \leq j \leq 3k-2$, $j-i+1 = |W|$ and for every $1 \leq a \leq k$, $f(v_a') = 2k-a$.
	\item\label{prop:oC2} All the vertices in $W$ have different images in $f$.
	\item\label{prop:oC3} If there exists a configuration $\mathsf{conf}^* = (w_1=v_1', w_2,\ldots,w_k=v_k')$ such that {\bf (1)} for every $1 \leq a \leq k$, $f(w_a) = 2k-a$, {\bf (2)} for every $\ell+1 \leq a \leq k$, $w_a = v_{a-\ell}$, and {\bf (3)} for every $r+1 \leq a \leq k$, $v_a'' = w_{a-r}$, then $(\mathsf{conf}, \mathsf{conf}^*)$ is an $\ell$-transition and $(\mathsf{conf}^*, \mathsf{conf}'')$ is an $r$-transition. 
\end{enumerate}
\end{lemma}

\begin{figure}
    \includegraphics[page=8,width=\textwidth]{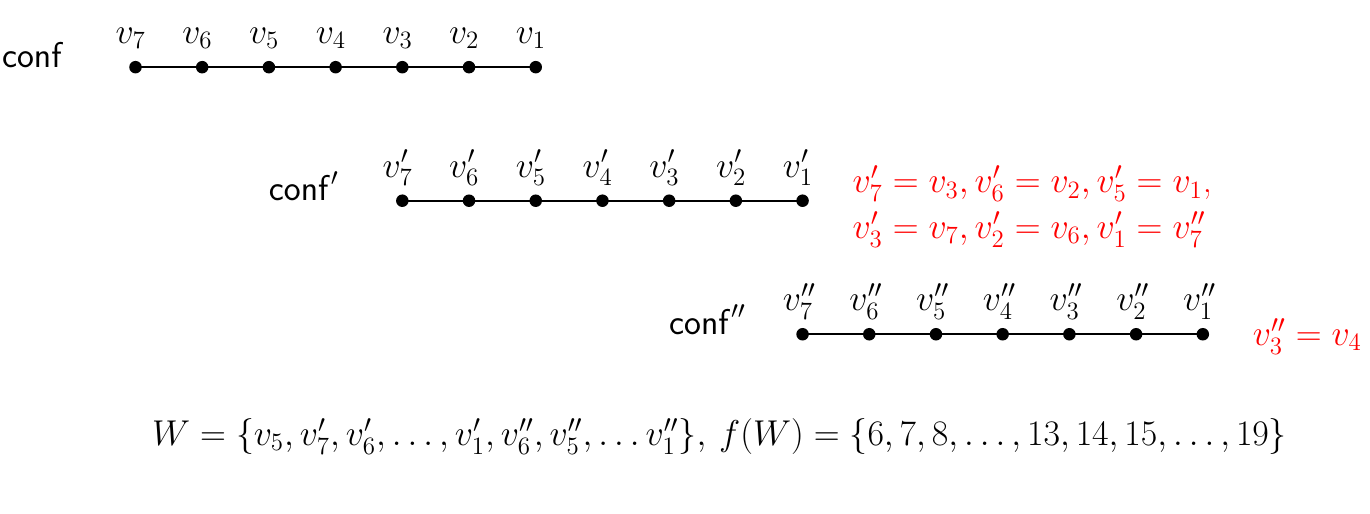}
    \caption{An illustration for Lemma~\ref{lem:orderedColoring} where $\ell = 4$ and $r = 6$.} 
  \label{fig:lemmaOrdered}
\end{figure}

\begin{proof}
We first prove Conditions~(\ref{prop:oC1}) and~(\ref{prop:oC2}). As $(\mathsf{conf},\mathsf{conf}')$ is an $\ell$-transition, by Lemma~\ref{lem:confGraph}, $\mathsf{conf}$ and $\mathsf{conf}'$ can have at most $\ell$ vertices in common. Similarly, $\mathsf{conf}'$ and $\mathsf{conf}''$ can have at most $r$ vertices in common. As $\ell, r \leq k-1$, from the definition of triplet order, $|W| \leq \ell + k + r \leq 3k-2$. Also $W \subseteq V(G)$, so by Definition~\ref{def:Permuter}, there exists a function in $\mathcal{G}$ that maps $W$ to an ordered set $\{i,i+1,\ldots,j\}$ for any $i, j \in \mathbb{N}$ such that $i \leq j \leq 3k-2$ and $j-i+1 = |W|$. As $|W| \geq k$, we can choose $i$ and $j$ such that there exists a function $g \in \mathcal{G}$ that maps $W$ to $\{i,i+1,\ldots,k,k+1,\ldots,2k-1,\ldots,j\}$ where for every $1 \leq a \leq k$, $g(v_a') = 2k-a$. It is easy to see that all the vertices in $W$ have different images in $g$.
Thus by taking $f = g$, Conditions~(\ref{prop:oC1}) and~(\ref{prop:oC2}) are satisfied. 

We now prove Condition~(\ref{prop:oC3}). Suppose that there exists a configuration $\mathsf{conf}^* = (w_1=v_1', w_2, \ldots,$ $w_k'=v_k')$ other than $\mathsf{conf}'$ such that {\bf (1)} for every $1 \leq a \leq k$, $f(w_a) = 2k-a$, {\bf (2)} for every $\ell+1 \leq a \leq k$, $w_a = v_{a-\ell}$, and {\bf (3)} for every $r+1 \leq a \leq k$, $v_a'' = w_{a-r}$. We prove that $(\mathsf{conf}, \mathsf{conf}^*)$ is an $\ell$-transition. The proof of $(\mathsf{conf}^*, \mathsf{conf}'')$ being an $r$-transition is similar.

By the way of contradiction, assume that $(\mathsf{conf}, \mathsf{conf}^*)$ is not an $\ell$-transition. Notice that for every $\ell+1 \leq a \leq k$, $w_a = v_{a-\ell}$. So, from Lemma~\ref{lem:confGraph}, there must exist a vertex $w_b$ such that $w_b = v_c$ for some $1 \leq b \leq \ell$ and $1 \leq c \leq (k+b)-(\ell+1)$. In particular, $f(v_c) = f(w_b) = 2k-b$. Moreover, $f(v_b') = 2k-b$. As $W$ contains every distinct vertex of $V(\mathsf{conf}) \cup V(\mathsf{conf}') \cup V(\mathsf{conf}'')$ exactly once, by Condition~(\ref{prop:oC2}), $v_b' = v_c$. As $(\mathsf{conf}, \mathsf{conf}')$ is an $\ell$-transition, from Lemma~\ref{lem:confGraph}, $v_b' \notin \{v_1,v_2,\ldots,v_{(k+b)-(\ell+1)}\}$. In particular, $v_b' \neq v_c$ which is a contradiction.
\end{proof}

We now give the following lemma, about finding a particular labeled path in a vertex labeled graph, which will be helpful throughout the section.

\begin{lemma}\label{lem:labPath}
Let $k, r, t \in \mathbb{N}$. Then, there exists an algorithm that, given an undirected graph $G$, a labeling function $f$ from $V(G)$ to $[t]$ and any two vertices $u$ and $v$ in $V(G)$, runs in time $\OO(|V(G)|+|E(G)|)$ and determines whether there exists a (simple) path $P =(w_1=u, w_2, \ldots, w_k=v)$ of length $k-1$ between $u$ and $v$ in $G$ such that for every $1 \leq i \leq k$, $f(w_i)= r+k-i$, and if the path exists, then it finds such a path.
\end{lemma}

\begin{proof}
We first construct a new directed graph $G'$ from $G$ as follows. Initialize $V(G') = V(G)$ and $E(G') = \emptyset$. For every edge $\{x,y\} \in E(G)$ such that $|f(x) - f(y)| = 1$: if $f(x) > f(y)$, then add the arc $(x, y)$ to $E(G')$; otherwise, add the arc $(y, x)$ to $E(G')$. It is easy to see that, if there exists a (simple) path $P =(w_1=u, w_2, \ldots, w_k=v)$ of length $k-1$ between $u$ and $v$ in $G$ such that for every $1 \leq i \leq k$, $f(w_i)= r+k-i$, then $P$ also exists in $G'$ as a directed path from $u$ to $v$. Note that in $G'$, every arc is directed from a vertex labeled $i+1$ to a vertex labeled $i$ for some $1 \leq i \leq t-1$ so if $f(u) = r+k-1$ and $f(v) = r$, then any path from $u$ to $v$ in $G'$ will satisfy the required property. So we can perform breadth-first search (BFS) in $G'$ starting from $u$ and return the shortest path from $u$ to $v$ in $G'$, if it exists. As we can construct $G'$ in time $\OO(|V(G)|+|E(G)|)$ and the BFS runs in $\OO(|V(G')| + |E(G')|)$, the lemma follows.
\end{proof}

Let $\mathbf{SG}=\langle G,k,\mathsf{init},\mathsf{fin}\rangle$ be an instance of {\sc Snake Game}. We now give a procedure (Algorithm~\ref{alg:sparseGraph}) to construct a new graph, called a {\em $(k-1)$-sparse configuration graph} of $\mathbf{SG}$, based on a $(V(G),3k-2)$-permuter and Lemma~\ref{lem:orderedColoring}. We first give an outline of the procedure. 

We initialize $\mathcal{C}'$ to be a new configuration graph whose vertex set contains $\mathsf{init}$ and $\mathsf{fin}$ and whose edge set is empty. Then, we construct a $(V(G),3k-2)$-permuter $\mathcal{G}$. For each function $g \in \mathcal{G}$, we take every pair of distinct vertices $(u,v)$ of $G$ and check whether there exists a path $P$ between $u$ and $v$ of length $k-1$ such that if we traverse the path $P$ from $u$ to $v$, the ordered set of images of the vertices of $P$ under $g$ is $\{2k-1, 2k-2, \ldots, k\}$. If such a path $P$ exists, which is checked using Lemma~\ref{lem:labPath}, then we add $P$ as a configuration with $u$ as its head to the vertex set of $\mathcal{C}'$. After finishing this for loop, we take every pair of distinct vertices $(\mathsf{conf},\mathsf{conf}')$ of $\mathcal{C}'$ and add an edge from $\mathsf{conf}$ to $\mathsf{conf}'$ if $(\mathsf{conf}$, $\mathsf{conf}')$ is a $(k-1)$-transition, which is checked using Observation~\ref{obs:testTransition}. In the end, we return $\mathcal{C}'$ and $\mathcal{G}$.

\begin{algorithm}
    \SetKwInOut{Input}{Input}
    \SetKwInOut{Output}{Output}
	\medskip
    {\textbf{function} sparseConfigurationGraph}$(\mathbf{SG},k)$\;
    	let $\mathcal{C}'$ be an empty graph\;
    	initialize $V(\mathcal{C}') = \{\mathsf{init}, \mathsf{fin}\}$\;
    	construct a $(V(G),3k-2)$-permuter $\mathcal{G}$ using Lemma~\ref{lem:sizePermuter}\;
	    \For{each $g \in \mathcal{G}$}{
	    	\For{each pair of vertices $(u,v) \in V(G) \times V(G)$ such that $u \neq v$}{
	    		\If{there is a (simple) path $P =(w_1=u, w_2, \ldots, w_k=v)$ of length $k-1$ between $u$ and $v$ in $G$ such that for every $1 \leq a \leq k$, $g(w_a)= 2k-a$}{\label{li:pathP} 
	    			add $P$ as a configuration with $u$ as its head to $V(\mathcal{C}')$\;
	    		}
	    	}
		}
		\For{every triplet of distinct vertices $(u, v, w) \in V(G) \times V(G) \times V(G)$}{
			\For{every pair of vertices $(\mathsf{conf} = (x_1=v, x_2, \ldots, x_k=u), \mathsf{conf}' = (x_1'=w, x_2', \ldots,$ $ x_k'=v)) \in V(\mathcal{C}') \times V(\mathcal{C}')$}{\label{li:pair} 
				\If{$(\mathsf{conf}, \mathsf{conf}')$ is a $(k-1)$-transition in $G$}{
					add the ordered pair $(\mathsf{conf}, \mathsf{conf}')$ as a directed edge in $E(\mathcal{C}')$;
				}
			}
		}
		\Return $(\mathcal{C}'$, $\mathcal{G})$.
    \caption{Construction of a $(k-1)$-sparse configuration graph.}
    \label{alg:sparseGraph}
\end{algorithm}

We define the notion of a {\em $(k-1)$-sparse configuration graph} as follows.

\begin{definition}
Let $\mathbf{SG}=\langle G,k,\mathsf{init},\mathsf{fin}\rangle$ be an instance of {\sc Snake Game}. Then, any graph returned by Algorithm~\ref{alg:sparseGraph} is called a {\em $(k-1)$-sparse configuration graph} of $\mathbf{SG}$.
\end{definition}

We first give the following lemma about the size of a $(k-1)$-sparse configuration graph and the running time of Algorithm~\ref{alg:sparseGraph}. 

\begin{lemma}\label{lem:algoTime}
Let $\mathbf{SG}=\langle G,k,\mathsf{init},\mathsf{fin}\rangle$ be an instance of {\sc Snake Game}. Then, Algorithm~\ref{alg:sparseGraph} runs in time $k^{\OO(k)}n^3\log^2 n$ and returns a $(k-1)$-sparse configuration graph of $\mathbf{SG}$ with $k^{\OO(k)} n^2\log n$ vertices and $k^{\OO(k)}n^3\log^2 n$ arcs.
\end{lemma}

\begin{proof}
Let $|V(G)| = n$, $|E(G)| = m$, $\mathcal{C}'$ be a $(k-1)$-sparse configuration graph and $\mathcal{G}$ be a $(V(G),3k-2)$-permuter constructed by Algorithm~\ref{alg:sparseGraph}. From Corollary~\ref{cor:edgeSetSize}, we implicitly assume that $m = k^{\OO(1)}n$. For each function $g \in \mathcal{G}$ and every pair of distinct vertices of $G$, the algorithm adds at most one vertex to $\mathcal{C}'$. Also for every triplet of distinct vertices of $G$, the algorithm adds an arc between every two vertices $\mathsf{conf}$ and $\mathsf{conf}'$ of $\mathcal{C}'$, such that the head of $\mathsf{conf}$ and tail of $\mathsf{conf}'$ is the same vertex and $(\mathsf{conf},\mathsf{conf}')$ is a $(k-1)$-transition. By Corollary~\ref{cor:(k-1)Transition}, we know that if the pair $(\mathsf{conf}, \mathsf{conf}')$ is a $(k-1)$-transition then the head of $\mathsf{conf}$ and tail of $\mathsf{conf}'$ should be the same vertex. As the algorithm performs this check for every triplet of distinct vertices of $G$, every pair $(\mathsf{conf},\mathsf{conf'}) \in V(\mathcal{C'}) \times V(\mathcal{C'})$ that is a $(k-1)$-transition is added to $\mathcal{C'}$. For every pair of distinct vertices of $G$, the algorithm adds at most $|\mathcal{G}|$ configurations so the number of pairs of vertices of $\mathcal{C'}$ satisfying the condition given in line~\ref{li:pair} is at most $|\mathcal{G}|^2$. By Lemma~\ref{lem:sizePermuter}, one can construct a $(V(G),t)$-permuter of size $t^{\OO(t)} \log n$ in time $t^{\OO(t)} n\log n$. So, $|V(\mathcal{C}')| = k^{\OO(k)} n^2\log n$ and $|E(\mathcal{C}')| = k^{\OO(k)} n^3\log^2 n$. By Observation~\ref{obs:testTransition}, for a given pair of configurations $(\mathsf{conf}, \mathsf{conf}')$, we can test whether $(\mathsf{conf}, \mathsf{conf}')$ is a $(k-1)$-transition in time $\OO(k^2)$. Also, by Lemma~\ref{lem:labPath}, we can find a path $P$ between $u$ and $v$ satisfying the condition given in line~\ref{li:pathP} of Algorithm~\ref{alg:sparseGraph} in time $\OO(n+m)$. So, the running time of Algorithm~\ref{alg:sparseGraph} is $(k^{\OO(k)} n\log n + (k^{\OO(k)} \log n) \cdot n^2 \cdot (n+m) + n^3 \cdot (k^{\OO(k)} \log^2 n) \cdot k^2) = k^{\OO(k)} n^3\log^2 n$. 
\end{proof}

Observe that, unlike the $(k-1)$-configuration graph, a $(k-1)$-sparse configuration graph may not be unique for a snake game $\mathbf{SG}$ even if the permuter is fixed. Indeed, given a function of the permuter and a pair of vertices of $G$, there may be more than one path that satisfies the condition given in line~\ref{li:pathP} of Algorithm~\ref{alg:sparseGraph} but we only add one of them (chosen arbitrarily) to our $(k-1)$-sparse configuration graph. 

The following lemma concerns the relationship between the $(k-1)$-configuration graph and a $(k-1)$-sparse configuration graph.

\begin{lemma}\label{lem:pgkToSparsek}
Let $\mathbf{SG}=\langle G,k,\mathsf{init},\mathsf{fin}\rangle$ be an instance of {\sc Snake Game}. Let $\mathcal{C}$ be the $(k-1)$-configuration graph and $\mathcal{C}'$ be a $(k-1)$-sparse configuration graph of $\mathbf{SG}$. Let $\mathsf{conf}, \mathsf{conf}', \mathsf{conf}'' \in V(\mathcal{C})$ such that both $(\mathsf{conf},\mathsf{conf}')$ and $(\mathsf{conf}',\mathsf{conf}'')$ are $(k-1)$-transitions. Then:
\begin{enumerate}[$(i)$]
	\item\label{prop:gk1} $\mathcal{C}'$ is an induced subgraph of $\mathcal{C}$.
	\item\label{prop:gk2} There exists a configuration $\mathsf{conf}^* \in V(\mathcal{C}')$ such that both $(\mathsf{conf},\mathsf{conf}^*)$ and $(\mathsf{conf}^*,\mathsf{conf}'')$ are also $(k-1)$-transitions.
\end{enumerate} 
\end{lemma}

\begin{proof}
We first prove that $V(\mathcal{C}') \subseteq V(\mathcal{C})$. As $V(\mathcal{C})$ is the set of all $(G,k)$-configurations, it suffices to show that every vertex in $V(\mathcal{C}')$ is a $(G,k)$-configuration. Clearly, $\mathsf{init}$ and $\mathsf{fin}$ are $(G,k)$-configurations. Let $\mathsf{conf}$ be a vertex in $V(\mathcal{C}')$ other than $\mathsf{init}$ and $\mathsf{fin}$. Then, by construction, $\mathsf{conf}$ is a simple path of length $k-1$, with a head and a tail, in $G$. So, by Definition~\ref{def:config}, $\mathsf{conf}$ is a $(G,k)$-configuration thus indeed, $V(\mathcal{C}') \subseteq V(\mathcal{C})$. We now prove that $\mathcal{C}'$ is an induced subgraph of $\mathcal{C}$. To this end, let $\mathsf{conf}$ and $\mathsf{conf}'$ be two configurations in $\mathcal{C}$. Then, there is an edge from $\mathsf{conf}$ to $\mathsf{conf}'$ in $\mathcal{C}$ if and only if $(\mathsf{conf},\mathsf{conf}')$ is a $(k-1)$-transition. Moreover, if $\mathsf{conf}$ and $\mathsf{conf}'$ belong to $\mathcal{C}'$ then the same condition is true in $\mathcal{C}'$ by construction. In turn, this proves that $\mathcal{C}'$ is indeed an induced subgraph of $\mathcal{C}$.

We now prove Condition~(\ref{prop:gk2}). Let $\mathsf{conf}=(v_1,v_2,\ldots,v_k)$, $\mathsf{conf}'=(v_1',v_2',\ldots,v_k' = v_1)$ and $\mathsf{conf}''=(v_1'',v_2'',\ldots,v_k'' = v_1')$. Let $W$ be the triplet order of $(\mathsf{conf},\mathsf{conf}',\mathsf{conf}'')$ and $\mathcal{G}$ be a $(V(G),3k-2)$-permuter. As both $(\mathsf{conf},\mathsf{conf}')$ and $(\mathsf{conf}',\mathsf{conf}'')$ are $(k-1)$-transitions, by Condition~(\ref{prop:oC1}) of Lemma~\ref{lem:orderedColoring}, there exists a function $f \in \mathcal{G}$ that maps $W$ to an ordered set $\{i,i+1,\ldots,k,k+1,\ldots,2k-1,\ldots,j\}$ such that $1 \leq i \leq j \leq 3k-2$, $j-i+1 = |W|$ and for every $1 \leq a \leq k$, $f(v_a') = 2k-a$. Let $\mathsf{prop}$ be the property of the path $P$ stated in line~\ref{li:pathP} of Algorithm~\ref{alg:sparseGraph}. So, when Algorithm~\ref{alg:sparseGraph} considers $g=f, u=v_1', v =v_k'$, $\mathsf{conf}'$ satisfies the property $\mathsf{prop}$. As there can be more than one path between $v_1'$ and $v_k'$ that satisfies the property $\mathsf{prop}$, Algorithm~\ref{alg:sparseGraph} either adds $\mathsf{conf}'$ or any other path between $v_1'$ and $v_k'$ satisfying the property $\mathsf{prop}$. If the algorithm adds $\mathsf{conf}'$, then $\mathsf{conf}^* = \mathsf{conf}'$ and we are done. 

Otherwise, let $\widehat{\mathsf{conf}}=(w_1=v_1',w_2,\ldots,w_k=v_k')$ be the configuration added by the algorithm. As $\widehat{\mathsf{conf}}$ satisfies the property $\mathsf{prop}$, for every $1 \leq a \leq k$, $g(w_a)= 2k-a$. Also, observe that $w_k = v_k' = v_1$ and $v_k'' = v_1' = w_1$, so by Condition~(\ref{prop:oC3}) of Lemma~\ref{lem:orderedColoring}, both $(\mathsf{conf},\widehat{\mathsf{conf}})$ and $(\widehat{\mathsf{conf}},\mathsf{conf}'')$ are $(k-1)$-transitions. So, we can take $\mathsf{conf}^* = \widehat{\mathsf{conf}}$ and we are done. 
\end{proof}
We now give the following lemma which relates paths in the $(k-1)$-configuration graph and a $(k-1)$-sparse configuration graph.

\begin{lemma}\label{lem:gkToSparsek}
Let $\mathbf{SG}=\langle G,k,\mathsf{init},\mathsf{fin}\rangle$ be an instance of {\sc Snake Game} and $t \in \mathbb{N}$. Let $\mathcal{C}$ be the $(k-1)$-configuration graph and $\mathcal{C}'$ be a $(k-1)$-sparse configuration graph of $\mathbf{SG}$. Let $\mathsf{conf}_s, \mathsf{conf}_e \in V(\mathcal{C})$. Then, there exists a path $P = (\mathsf{conf}_1 = \mathsf{conf}_s, \mathsf{conf}_2, \ldots, \mathsf{conf}_{t+1} = \mathsf{conf}_e)$ of length $t$ from $\mathsf{conf}_s$ to $\mathsf{conf}_e$ in $\mathcal{C}$ if and only if there exists a path $P' = (\mathsf{conf}'_1 = \mathsf{conf}_s, \mathsf{conf}'_2, \ldots, \mathsf{conf}'_t, \mathsf{conf}'_{t+1} = \mathsf{conf}_e)$ of length $t$ from $\mathsf{conf}_s$ to $\mathsf{conf}_e$ in $\mathcal{C}$ such that $\mathsf{conf}'_i \in V(\mathcal{C}')$ for every $2 \leq i \leq t$.
\end{lemma}

\begin{proof}
We first prove the if part. The if part is straightforward: from Condition~(\ref{prop:gk1}) of Lemma~\ref{lem:pgkToSparsek}, $V(\mathcal{C}') \subseteq V(\mathcal{C})$ and $E(\mathcal{C}') \subseteq E(\mathcal{C})$, so take $P=P'$.

We now prove the only if part. Let $P = (\mathsf{conf}_1 = \mathsf{conf}_s, \mathsf{conf}_2, \ldots, \mathsf{conf}_{t+1} = \mathsf{conf}_e)$ be a path of length $t$ from $\mathsf{conf}_s$ to $\mathsf{conf}_e$ in $\mathcal{C}$. As $P$ is a path in $\mathcal{C}$ so, for every $1 \leq i \leq t$, $(\mathsf{conf}_i, \mathsf{conf}_{i+1})$ is a $(k-1)$-transition. Now consider the triplet $(\mathsf{conf}_1, \mathsf{conf}_2, \mathsf{conf}_3)$. From Condition~(\ref{prop:gk2}) of Lemma~\ref{lem:pgkToSparsek}, there exists a configuration $\mathsf{conf}'_2 \in V(\mathcal{C}')$ such that both $(\mathsf{conf}_1,\mathsf{conf}'_2)$ and $(\mathsf{conf}'_2,\mathsf{conf}_3)$ are $(k-1)$-transitions. Now we construct a new path $P^*$ from $P$ by replacing $\mathsf{conf}_2$ with $\mathsf{conf}'_2$. It is easy to see that $P^*$ is also a path in $\mathcal{C}$. So, by iteratively taking triplet $(\mathsf{conf}_i, \mathsf{conf}_{i+1},\mathsf{conf}_{i+2})$ for every $1 \leq i \leq (t-1)$ and applying this transformation, we get a new path $P' = (\mathsf{conf}'_1 = \mathsf{conf}_s, \mathsf{conf}'_2, \ldots, \mathsf{conf}'_{t+1} = \mathsf{conf}_e)$ of length $t$ in $\mathcal{C}$ such that $\mathsf{conf}'_i \in V(\mathcal{C}')$ for every $2 \leq i \leq t$. 
\end{proof}

By putting $\mathsf{conf}_s = \mathsf{init}, \mathsf{conf}_e = \mathsf{fin}$ in Lemma~\ref{lem:gkToSparsek}, we have the following simple corollary. 

\begin{corollary}\label{cor:kToSparsek}
Let $\mathbf{SG}=\langle G,k,\mathsf{init},\mathsf{fin}\rangle$ be an instance of {\sc Snake Game} and $t \in \mathbb{N}$. Let $\mathcal{C}$ be the $(k-1)$-configuration graph and $\mathcal{C}'$ be a $(k-1)$-sparse configuration graph of $\mathbf{SG}$. Then, there exists a path of length $t$ from $\mathsf{init}$ to $\mathsf{fin}$ in $\mathcal{C}$ if and only if there exists a path of length $t$ from $\mathsf{init}$ to $\mathsf{fin}$ in $\mathcal{C}'$
\end{corollary}

Given an instance of {\sc Snake Game} $\mathbf{SG}=\langle G,k,\mathsf{init},\mathsf{fin}\rangle$, we now show how we use its $(k-1)$-configuration graph to solve it. By Observation~\ref{obs:confGraph}, we know that $\mathbf{SG}$ is a \yes-instance if and only if there exists a path from $\mathsf{init}$ to $\mathsf{fin}$ in the $1$-configuration graph. It is easy to see that every path of length $(k-1)$ between two configurations in the $1$-configuration graph is compressed to a path of length $1$ in the $(k-1)$-configuration graph. If the answer to an instance of {\sc Snake Game} is \yes, then there may or may not exist a path from $\mathsf{init}$ to $\mathsf{fin}$ in the $1$-configuration graph whose length is a multiple of $(k-1)$. We first consider the case where the length of the path is a multiple of $(k-1)$ (Section~\ref{sec:multiple}), and then extend the result for the case where the length of the path is not a multiple of $(k-1)$ (Section~\ref{sec:notMultiple}).

\subsection{When the Length of the Path is a multiple of $(k-1)$}\label{sec:multiple}

In this subsection, we consider the case where the length of the path from $\mathsf{init}$ to $\mathsf{fin}$ in the $1$-configuration graph is a multiple of $(k-1)$. We give the following lemma which relates paths in the $1$-configuration graph and the $(k-1)$-configuration graph whose lengths are multiples of $(k-1)$.

\begin{lemma}\label{lem:g1Tok}
Let $\mathbf{SG}=\langle G,k,\mathsf{init},\mathsf{fin}\rangle$ be an instance of {\sc Snake Game} and $t \in \mathbb{N}$. Let $\mathcal{C}$ be the $(k-1)$-configuration graph of $\mathbf{SG}$. Let $\mathsf{conf}_s$ and $\mathsf{conf}_e$ be any two configurations of $G$. Then, there exists a path of length $t(k-1)$ from $\mathsf{conf}_s$ to $\mathsf{conf}_e$ in the $1$-configuration graph if and only if there exists a path of length $t$ from $\mathsf{conf}_s$ to $\mathsf{conf}_e$ in $\mathcal{C}$.
\end{lemma}

\begin{proof}
$(\Rightarrow)$ Suppose that there exists a path $P = (\mathsf{conf}_1 = \mathsf{conf}_s, \ldots, \mathsf{conf}_k, \ldots, \mathsf{conf}_{2k-1}, \ldots,$ $\mathsf{conf}_{t(k-1)+1} = \mathsf{conf}_e)$ of length $t(k-1)$ from $\mathsf{conf}_s$ to $\mathsf{conf}_e$ in the $1$-configuration graph. For every $i = 1 + r(k-1)$ where $0 \leq r \leq (t-1)$, we replace the subpaths $(\mathsf{conf}_i, \mathsf{conf}_{i+1}, \ldots, \mathsf{conf}_{i+k-1})$ of length $(k-1)$ in $P$ with $(\mathsf{conf}_i, \mathsf{conf}_{i+k-1})$ iteratively to obtain a new path $P' = (\mathsf{conf}'_1=\mathsf{conf}_s, \mathsf{conf}'_2 = \mathsf{conf}_k, \mathsf{conf}'_3 = \mathsf{conf}_{2k-1}, \ldots, \mathsf{conf}'_{t+1}=\mathsf{conf}_e)$. Observe that every pair $(\mathsf{conf}'_i,$ $\mathsf{conf}'_{i+1})$ in $P'$ is a $(k-1)$-transition, so $P'$ is a path of length $t$ from $\mathsf{conf}_s$ to $\mathsf{conf}_e$ in $\mathcal{C}$.

\noindent$(\Leftarrow)$ Suppose that there exists a path $P = (\mathsf{conf}_1=\mathsf{conf}_s,\mathsf{conf}_2,\ldots,\mathsf{conf}_{t+1}=\mathsf{conf}_e)$ of length $t$ from $\mathsf{conf}_s$ to $\mathsf{conf}_e$ in $\mathcal{C}$. As every arc $(\mathsf{conf}_i,\mathsf{conf}_{i+1})$ in path $P$ is a $(k-1)$-transition, we can replace every arc in $P$ with the corresponding path of length $(k-1)$. By applying this transformation, we get a new path $P' = (\mathsf{conf}_1' = \mathsf{conf}_s, \ldots, \mathsf{conf}_k' = \mathsf{conf}_2, \ldots, \mathsf{conf}_{2k-1}' = \mathsf{conf}_3, \ldots,$ $\mathsf{conf}_{t(k-1)+1}' = \mathsf{conf}_e)$ of length $t(k-1)$ where every pair $(\mathsf{conf}_i',\mathsf{conf}_{i+1}')$ in $P'$ is a $1$-transition, so $P'$ is a path from $\mathsf{conf}_s$ to $\mathsf{conf}_e$ in the $1$-configuration graph. 
\end{proof}

By putting $\mathsf{conf}_s = \mathsf{init}$ and $\mathsf{conf}_e = \mathsf{fin}$ in Lemma~\ref{lem:g1Tok}, we have the following simple corollary. 

\begin{corollary}\label{cor:1Tok}
Let $\mathbf{SG}=\langle G,k,\mathsf{init},\mathsf{fin}\rangle$ be an instance of {\sc Snake Game} and $t \in \mathbb{N}$. Let $\mathcal{C}$ be the $(k-1)$-configuration graph of $\mathbf{SG}$. Then, there exists a path of length $t(k-1)$ from $\mathsf{init}$ to $\mathsf{fin}$ in the $1$-configuration graph if and only if there exists a path of length $t$ from $\mathsf{init}$ to $\mathsf{fin}$ in $\mathcal{C}$.
\end{corollary}

The combination of Corollaries~\ref{cor:kToSparsek} and~\ref{cor:1Tok} gives us the following observation.

\begin{observation}\label{obs:1ToSparsek}
Let $\mathbf{SG}=\langle G,k,\mathsf{init},\mathsf{fin}\rangle$ be an instance of {\sc Snake Game} and $t \in \mathbb{N}$. Let $\mathcal{C}'$ be a $(k-1)$-sparse configuration graph of $\mathbf{SG}$. Then, there exists a path of length $t(k-1)$ from $\mathsf{init}$ to $\mathsf{fin}$ in the $1$-configuration graph if and only if there exists a path of length $t$ from $\mathsf{init}$ to $\mathsf{fin}$ in $\mathcal{C}'$.
\end{observation}

\subsection{When the Length of the Path is not a multiple of $(k-1)$}\label{sec:notMultiple}

We now consider the case where the length of the path from $\mathsf{init}$ to $\mathsf{fin}$ in the $1$-configuration graph is not a multiple of $(k-1)$. Let $r \in \mathbb{N}$ such that $r < k-1$. Let $\mathcal{C}$ be the $(k-1)$-configuration graph and $\mathcal{C}'$ be a $(k-1)$-sparse configuration graph of $\mathbf{SG}$. Let $P = (\mathsf{conf}_1=\mathsf{init},\ldots,\mathsf{conf}_k,\ldots,\mathsf{conf}_{2k-1},\ldots,\mathsf{conf}_{t(k-1)+(r+1)}=\mathsf{fin})$ be a path from $\mathsf{init}$ to $\mathsf{fin}$ in the $1$-configuration graph. Note that the length of the path $P$ is $t(k-1) + r$, which is not a multiple of $k-1$. Consider the subpath $P' = (\mathsf{conf}_1=\mathsf{init},\ldots,\mathsf{conf}_k,\ldots,\mathsf{conf}_{2k-1},\ldots,$ $\mathsf{conf}_{t(k-1)+1})$ of $P$ in the $1$-configuration graph. By Lemmas~\ref{lem:gkToSparsek} and~\ref{lem:g1Tok}, there exists a path $P'' = (\mathsf{conf}'_1=\mathsf{init},\mathsf{conf}'_2, \mathsf{conf}'_3, \ldots, \mathsf{conf}'_{t+1} = \mathsf{conf}_{t(k-1)+1})$ of length $t$ in $\mathcal{C}$ such that $\mathsf{conf}'_i \in V(\mathcal{C}')$ for every $1 \leq i \leq t$.

Now, consider the path $P^* = (\mathsf{conf}'_1 = \mathsf{init}, \mathsf{conf}'_2, \ldots, \mathsf{conf}'_{t+1}, \mathsf{conf}'_{t+2} = \mathsf{fin})$ defined on the vertex set $V = V(P'') \cup \{\mathsf{fin}\}$. Observe that, in $P^*$, all the vertices except possibly $\mathsf{conf}'_{t+1}$ belong to $V(\mathcal{C}')$. For every $1 \leq i \leq t$, $(\mathsf{conf}'_i,\mathsf{conf}'_{i+1})$ is a $(k-1)$-transition, but possibly $\mathsf{conf}'_{t+1} \notin V(\mathcal{C}')$, we only know that $(\mathsf{conf}'_i,\mathsf{conf}'_{i+1}) \in E(\mathcal{C}')$ for every $1 \leq i \leq t-1$. Moreover, $(\mathsf{conf}'_{t+1},\mathsf{conf}'_{t+2} = \mathsf{fin})$ is an $r$-transition so $(\mathsf{conf}'_{t+1}, \mathsf{fin}) \notin E(\mathcal{C}')$. To handle the vertex $\mathsf{conf}'_{t+1}$ and edges $(\mathsf{conf}'_t,\mathsf{conf}'_{t+1})$ and $(\mathsf{conf}'_{t+1},\mathsf{conf}'_{t+2} = \mathsf{fin})$, we now give a procedure (Algorithm~\ref{alg:gSparseGraph}) to construct a new graph called a {\em generalized $(k-1)$-sparse configuration graph} of $\mathbf{SG}$ based on a $(k-1)$-sparse configuration graph. We first give an outline of the procedure.

We first initialize $\mathcal{C}'$ to be a $(k-1)$-sparse configuration graph of $\mathbf{SG}$ and $\mathcal{G}$ to be a corresponding $(V(G),3k-2)$-permuter, returned by Algorithm~\ref{alg:sparseGraph}. Then, we initialize a new graph $\mathcal{C}''$ as $\mathcal{C}'$, i.e. all the vertices and edges of $\mathcal{C}'$ are present in $\mathcal{C}''$. In addition to these vertices and edges, for every function $g \in \mathcal{G}$ and every pair of vertices $(u,v)$ of $G$, such that $v \neq u$ and $u = f_{s+1}$ for some $1 \leq s \leq k-2$, we add at most a vertex and an edge as follows. If there exists a path $\widehat{P}$ of length $k-1$ between $u$ and $v$ such that $(\widehat{P}, \mathsf{fin})$ is a $s$-transition, considering $\widehat{P}$ as a configuration with $u$ as its head and while traversing the path $\widehat{P}$ from $u$ to $v$, the ordered set of images of the vertices of $\widehat{P}$ under $g$ is $\{2k-1, 2k-1, \ldots, k\}$, then we add $\widehat{P}$ as a configuration with $u$ as its head to $V(\mathcal{C}'')$ and the ordered pair $(\widehat{P}, \mathsf{fin})$ as a directed edge to $E(\mathcal{C}'')$. In the end, we return $\mathcal{C}''$.

\begin{algorithm}
    \SetKwInOut{Input}{Input}
    \SetKwInOut{Output}{Output}
	\medskip
    {\textbf{function} generalizedSparseConfigurationGraph}$(\mathbf{SG},k)$\;
    	let $\mathsf{fin} = (f_1,f_2,\ldots,f_k)$\;
    	initialize $(\mathcal{C}', \mathcal{G})$ = sparseConfigurationGraph$(\mathbf{SG},k)$\;
    	initialize $\mathcal{C}'' = \mathcal{C}'$\;  
	    \For{each $g \in \mathcal{G}$}{
	    	\For{every vertex $u \in V(G)$ such that $u = f_{s+1}$ for some $1 \leq s \leq k-2$}{
		    	\For{every vertex $v \in V(G)$ such that $v \neq u$}{
		    		\If{there is a (simple) path $\widehat{P} = (w_1=u, $ $w_2, \ldots,w_k=v)$ of length $k-1$ between $u$ and $v$ in $G$ such that, for every $1 \leq a \leq k$, $g(w_a)= 2k-a$ and for every $1 \leq a \leq k-s$, $w_a = f_{a+s}$}{
		    			\label{li:pathPe}
		    			add $\widehat{P}$ as a configuration with $u$ as its head to $V(\mathcal{C}'')$\;
		    			add ordered pair $(P, \mathsf{fin})$ as a directed edge in $E(\mathcal{C}'')$;
		    		}
		    	}
	    	}
	    	
		}
		\Return $\mathcal{C}''$.
    \caption{Construction of a generalized $(k-1)$-sparse configuration graph.}
    \label{alg:gSparseGraph}
\end{algorithm}

Similarly to the notion of a $(k-1)$-sparse configuration graph, we now define the notion of a {\em generalized $(k-1)$-sparse configuration graph} as follows.

\begin{definition}
Let $\mathbf{SG}=\langle G,k,\mathsf{init},\mathsf{fin}\rangle$ be an instance of {\sc Snake Game}. Then, any graph returned by Algorithm~\ref{alg:gSparseGraph} is called a {\em generalized $(k-1)$-sparse configuration graph}.
\end{definition}

We give the following lemma which relates paths in the $1$-configuration graph and a generalized $(k-1)$-sparse configuration graph whose lengths are not multiples of $(k-1)$.

\begin{lemma}\label{lem:1TogSparsek}
Let $\mathbf{SG}=\langle G,k,\mathsf{init},\mathsf{fin}\rangle$ be an instance of {\sc Snake Game} and $t, r \in \mathbb{N}$ such that $r < k-1$. Let $\mathcal{C}''$ be a generalized $(k-1)$-sparse configuration graph of $\mathbf{SG}$. Then, there exists a path of length $t(k-1) + r$ from $\mathsf{init}$ to $\mathsf{fin}$ in the $1$-configuration graph if and only if there exists a path of length $t+1$ from $\mathsf{init}$ to $\mathsf{fin}$ in $\mathcal{C}''$.
\end{lemma}

\begin{proof}
Let $\mathcal{C}$ be the $(k-1)$-configuration graph of $\mathbf{SG}$. Let $P = (\mathsf{conf}_1=\mathsf{init},\ldots,\mathsf{conf}_k,\ldots,$ $\mathsf{conf}_{2k-1},\ldots,\mathsf{conf}_{t(k-1)+(r+1)}=\mathsf{fin})$ be a path from $\mathsf{init}$ to $\mathsf{fin}$ in the $1$-configuration graph. From the above discussion, we can get a path $P^* = (\mathsf{conf}'_1 = \mathsf{init}, \mathsf{conf}'_2, \ldots, \mathsf{conf}'_{t+1}, \mathsf{conf}'_{t+2} = \mathsf{fin})$ defined on the vertex set $V = V(P'') \cup \{\mathsf{fin}\}$ such that all the vertices except possibly $\mathsf{conf}'_{t+1}$ belong to $V(\mathcal{C}')$ and all the edges except possibly $(\mathsf{conf}'_t,\mathsf{conf}'_{t+1})$ and $(\mathsf{conf}'_{t+1},\mathsf{conf}'_{t+2} = \mathsf{fin})$ belong to $E(\mathcal{C}')$.

Observe that $\mathcal{C}'$ is a subgraph of $\mathcal{C}''$. So, if we can find a configuration $\mathsf{conf}^* \in V(\mathcal{C}'')$ such that $(\mathsf{conf}'_t, \mathsf{conf}^*)$ is a $(k-1)$-transition, $(\mathsf{conf}^*, \mathsf{conf}'_{t+2})$ is an $r$-transition, and $(\mathsf{conf}'_t, \mathsf{conf}^*)$ and $(\mathsf{conf}^*, \mathsf{conf}'_{t+2})$ belong to $E(\mathcal{C}'')$, then we can replace $\mathsf{conf}'_{t+1}$ with $\mathsf{conf}^*$ in $P^*$ to get a path from $\mathsf{init}$ to $\mathsf{fin}$ in $\mathcal{C}''$. 

Consider the triplet $(\mathsf{conf}'_t, \mathsf{conf}'_{t+1}, \mathsf{conf}'_{t+2} = \mathsf{fin})$. Let $\mathsf{conf}'_{t+1} = (v_1', v_2', \ldots, v_k')$ and $\mathsf{fin} = (f_1,f_2,\ldots,f_k)$. Let $W$ be the triplet order of $(\mathsf{conf}'_t, \mathsf{conf}'_{t+1}, \mathsf{fin})$ and $\mathcal{G}$ be a $(V(G),3k-2)$-permuter. As $(\mathsf{conf}'_t, \mathsf{conf}'_{t+1})$  is a $(k-1)$-transition and $(\mathsf{conf}'_{t+1}, \mathsf{fin})$ is an $r$-transition ($r < k-1$), by Condition~(\ref{prop:oC1}) of Lemma~\ref{lem:orderedColoring}, there exists a function $h \in \mathcal{G}$ that maps $W$ to an ordered set $\{i,i+1,\ldots,k,k+1,\ldots,2k-1,\ldots,j\}$ such that $1 \leq i \leq j \leq 3k-2$, $j-i+1 = |W|$ and for every $1 \leq a \leq k$, $h(v_a') = 2k-a$. Also, as $(\mathsf{conf}'_{t+1}, \mathsf{fin})$ is an $r$-transition ($r < k-1$), by Lemma~\ref{lem:confGraph}, $v_1' = f_{r+1}$. Let $\mathsf{prop}$ be the property of the path $\widehat{P}$ stated in line~\ref{li:pathPe} of Algorithm~\ref{alg:gSparseGraph}. So, when Algorithm~\ref{alg:gSparseGraph} considers $g=h, u=v_1'$ and $v =v_k'$, $\mathsf{conf}'_{t+1}$ satisfies the property $\mathsf{prop}$ for $s = r$. But there can be more than one path between $v_1'$ and $v_k'$ that satisfy the property $\mathsf{prop}$, so Algorithm~\ref{alg:gSparseGraph} either adds $\mathsf{conf}'_{t+1}$ or any other path between $v_1'$ and $v_k'$ satisfying the property $\mathsf{prop}$. If the algorithm adds $\mathsf{conf}'_{t+1}$, then $\mathsf{conf}^* = \mathsf{conf}'_{t+1}$ and we are done.

Otherwise, let $\widehat{\mathsf{conf}}=(w_1=v_1',w_2,\ldots,w_k=v_k')$ be the configuration added by the algorithm. As $\widehat{\mathsf{conf}}$ satisfies the property $\mathsf{prop}$ for $s = r$ so for every $1 \leq a \leq k$, $g(w_a)= 2k-a$ and for every $1 \leq a \leq k-r$, $w_a = f_{a+r}$. Additionally, observe that $w_k = v_k' = v_1$, by Condition~(\ref{prop:oC3}) of Lemma~\ref{lem:orderedColoring}, $(\mathsf{conf},\widehat{\mathsf{conf}})$ is a $(k-1)$-transition and $(\widehat{\mathsf{conf}},\mathsf{conf}'')$ is an $r$-transition. So, we can take $(\mathsf{conf}^* = \widehat{\mathsf{conf}})$ and we are done.
\end{proof}

As $\mathcal{C}'$ is a subgraph of $\mathcal{C}''$, combining Observation~\ref{obs:1ToSparsek} and Lemma~\ref{lem:1TogSparsek} results in the following theorem.

\begin{theorem}\label{the:1TogSparse}
Let $\mathbf{SG}=\langle G,k,\mathsf{init},\mathsf{fin}\rangle$ be an instance of {\sc Snake Game} and $t, r \in \mathbb{N}$ such that $r \leq k-1$. Then, there exists a path of length $t(k-1) + r$ from $\mathsf{init}$ to $\mathsf{fin}$ in the $1$-configuration graph if and only if there exists a path of length $t+1$ from $\mathsf{init}$ to $\mathsf{fin}$ in any generalized $(k-1)$-sparse configuration graph.
\end{theorem}

We now give the following lemma about the size of a generalized $(k-1)$-sparse configuration graph and the running time of Algorithm~\ref{alg:gSparseGraph}. 

\begin{lemma}\label{lem:gAlgoTime}
Let $\mathbf{SG}=\langle G,k,\mathsf{init},\mathsf{fin}\rangle$ be an instance of {\sc Snake Game}. Then, Algorithm~\ref{alg:gSparseGraph} runs in time $k^{\OO(k)} n^3\log^2 n$ and returns a generalized $(k-1)$-sparse configuration graph of $\mathbf{SG}$ with $k^{\OO(k)} n^2\log n$ vertices and $k^{\OO(k)} n^3\log^2 n$ arcs.
\end{lemma}

\begin{proof}
Let $|V(G)| = n$, $|E(G)| = m$, $\mathcal{C}'$ be a $(k-1)$-sparse configuration graph and $\mathcal{G}$ be a $(V(G),3k-2)$-permuter constructed by Algorithm~\ref{alg:sparseGraph}. From Corollary~\ref{cor:edgeSetSize}, we implicitly assume that $m = k^{\OO(1)}n$. Let $\mathcal{C}''$ be a generalized $(k-1)$-sparse configuration graph. In addition to the vertices of $\mathcal{C}'$, Algorithm~\ref{alg:gSparseGraph} adds at most a vertex and an edge to $\mathcal{C}''$ for each function $g \in \mathcal{G}$ and every pair of distinct vertices in $G$. From Lemma~\ref{lem:algoTime}, we know that $|V(\mathcal{C}')| = k^{\OO(k)} n^2\log n$ and $|E(\mathcal{C}')| = k^{\OO(k)} n^3\log^2 n$ and Algorithm~\ref{alg:sparseGraph} runs in time $k^{\OO(k)} n^3\log^2 n$. So, $|V(\mathcal{C}'')| = k^{\OO(k)} n^2\log n$ and $|E(\mathcal{C}'')| = k^{\OO(k)} n^3\log^2 n$. Also, again by Lemma~\ref{lem:labPath}, we can find a path $\widehat{P}$ between $u$ and $v$ satisfying the the condition given in line~\ref{li:pathPe} of Algorithm~\ref{alg:gSparseGraph} in time $\OO(n+m)$. So, the running time of Algorithm~\ref{alg:gSparseGraph} is $k^{\OO(k)} n^3\log^2 n + (k^{\OO(k)} \log n) \cdot n^2 \cdot (n+m)) = k^{\OO(k)} n^3\log^2 n$. 
\end{proof}

Thus, from Theorem~\ref{the:1TogSparse} and Lemma~\ref{lem:gAlgoTime}, we derive the following theorem. 

\begin{theorem}
There exists an algorithm that, given an instance of {\sc Snake Game} $\mathbf{SG}=\langle G,k,\mathsf{init},\mathsf{fin}\rangle$, solves $\mathbf{SG}$ in time $k^{\OO(k)}|V(G)|^3\log^2 |V(G)|$. Moreover, it finds the shortest path from $\mathsf{init}$ to $\mathsf{fin}$ if one exists.
\end{theorem}

\section{Non-Existence of a Polynomial Kernel on Grid Graphs}\label{sec:kernel}

In this section, we prove that the {\sc Snake Game} problem is unlikely to admit a polynomial kernel even on grid graphs. Our proof is based on the exhibition of a cross-composition (see Definition \ref{def:cross-comp}). We remind that our cross-composition is inspired by the proof of \NP-hardness of {\sc Snake Game} on grid graphs by Biasi and Ophelders \cite{DBLP:journals/tcs/BiasiO18}. As the source problem $\Pi$ for the cross-composition, we select the {\sc Hamiltonian Cycle} problem on grid graphs, and the target problem is {\sc Snake Game} on grid graphs. In what follows, we first present the reduction that serves as the cross-composition (Section \ref{sec:kernelReduction}), and then prove its correctness (Section \ref{sec:kernelReductionCorrectness}).

\subsection{Construction of the Reduction Function}\label{sec:kernelReduction}

We refer to our reduction function as $\mathsf{HamToSna}$. Its input consists of $t$ instances $I_1,I_2,\ldots,I_t$ of the {\sc Hamiltonian Cycle} problem on grid graphs where $t\in\mathbb{N}$, and its output is a single instance of the {\sc Snake Game} problem on grid graphs. Each instance $I_i$ is a grid graph $G_i$ with $n$ vertices, which we can assume without loss of generality to be connected else the instance is trivially a \no-instance. We need the following definition for describing the reduction function.

For the construction of $\mathsf{HamToSna}$, we would like to ``separate'' the $t$ input grid graphs in the sense that we would like to identify specific ``squares'' in the output grid graph where the input grid graphs will be embedded. The notion of a {\em boundary square of a grid graph} (defined below) will serve this purpose. 
As intuition for its definition and the next observation, note that any connected grid graph $G$ can be thought of as an induced subgraph of a $(|V(G)|\times |V(G)|)$-$0$- grid. We will denote the four corner vertices of this $(|V(G)|\times |V(G)|)$-$0$-grid by $(r_{\mathrm{min}},c_{\mathrm{min}}), (r_{\mathrm{min}},c_{\mathrm{max}}),(r_{\mathrm{max}},c_{\mathrm{min}})$ and $(r_{\mathrm{max}},c_{\mathrm{max}})$. Accordingly, we will say that $G$ is ``bounded'' by $(r_{\mathrm{min}},r_{\mathrm{max}},c_{\mathrm{min}},c_{\mathrm{max}})$. We would like this boundary to be unique, to which end we have Conditions \ref{cond:1}, \ref{cond:2} and \ref{cond:3} as follows.   

\begin{definition} [{\bf A $(r_{\mathrm{min}},r_{\mathrm{max}},c_{\mathrm{min}},c_{\mathrm{max}})$-Boundary Square of a Connected Grid Graph}]\label{def:Square grid graph}
A {\em boundary square} of a connected grid graph $G$ is a quadruple $(r_{\mathrm{min}},r_{\mathrm{max}},c_{\mathrm{min}},c_{\mathrm{max}})$ where $r_{\mathrm{min}},r_{\mathrm{max}},c_{\mathrm{min}},c_{\mathrm{max}}\in \mathbb{N}_0$, such that the following conditions are satisfied.
\begin{enumerate}
\item $r_{\mathrm{min}}=\mathrm{min}\{ r\in \mathbb{N}_0~|~$there exists $c\in \mathbb{N}_0$ such that $(r,c)\in V(G) \}$.\label{cond:1} 
\item $c_{\mathrm{min}}=\mathrm{min}\{ c\in \mathbb{N}_0~|~$there exists $r\in \mathbb{N}_0$ such that $(r,c)\in V(G) \}$.\label{cond:2}
\item $r_{\mathrm{max}}-r_{\mathrm{min}}=c_{\mathrm{max}}-c_{\mathrm{min}}=|V(G)|-1$.\label{cond:3}
\end{enumerate}
\end{definition}

\begin{observation} \label{obs:squreExists}
Let $G$ be a connected grid graph. Then, there exists a quadruple $(r_{\mathrm{min}},r_{\mathrm{max}},c_{\mathrm{min}},c_{\mathrm{max}})\in \mathbb{N}_0^4$ such that $(r_{\mathrm{min}},r_{\mathrm{max}},c_{\mathrm{min}},c_{\mathrm{max}})$ is a boundary square of $G$, and it is unique.
\end{observation}  

\paragraph{The $\mathsf{HamToSna}$ Reduction Function: Preprocessing.}
Let $G_1,G_2,\ldots,G_t$ be the $t$ instances of {\sc Hamiltonian Cycle} on grid graphs. From Observation \ref{obs:squreExists}, we know that there exists a boundary square of $G_i$ for every $1\leq i\leq t$. Let $(\widehat{r}^i_{\mathrm{min}},\widehat{r}^i_{\mathrm{max}},\widehat{c}^i_{\mathrm{min}},\widehat{c}^i_{\mathrm{max}})\in \mathbb{N}_0^4$ be the boundary square of $G_i$. Then, we modify $G_i$ as follows. For every vertex $(a,b)\in V(G_i)$, we replace it by the vertex $(a - \widehat{r}^i_{\mathrm{min}} + (n+1), b - \widehat{c}^i_{\mathrm{min}} + i(n+1))$. Intuitively, this operation simply means that we ``shift'' the entire grid graph $G_i$ in parallel to the axes. It is easy to see that the new boundary square $(r_{\mathrm{min}}^i,r_{\mathrm{max}}^i,c_{\mathrm{min}}^i,c_{\mathrm{max}}^i)\in \mathbb{N}_0^4$ of $G_i$ satisfies the following conditions.

\begin{itemize}
\item $r_{\mathrm{min}}^i=n+1$.
\item $c_{\mathrm{min}}^i=i(n+1)$.
\end{itemize}

In addition, because $(r_{\mathrm{min}}^i,r_{\mathrm{max}}^i,c_{\mathrm{min}}^i,c_{\mathrm{max}}^i)$ is the boundary square of $G_i$, Condition \ref{cond:1} in Definition \ref{def:Square grid graph} implies that there exists a vertex $(a,b)\in V(G_i)$ such that $a=r_{\mathrm{min}}^i$. As there can be more than one such vertex, for every $1\leq i\leq t$, we arbitrarily choose one such vertex in $V(G_i)$ and denote it by $(r_{\mathrm{min}}^i,c_i)$.

\paragraph{The $\mathsf{HamToSna}$ Reduction Function: Construction} After executing the above preprocessing step, we define $\mathsf{HamToSna}(G_1,G_2,\ldots,G_t) = \langle G,n,\mathsf{init},\mathsf{fin}\rangle$ as follows. (See Figure~\ref{fig:noKernel}.
\begin{enumerate}
\item$V(G)=V^{\mathsf{InitFinConfs}}\cup V^{\mathsf{Connectors}}\cup V^{\mathsf{Instances}}$, where
\begin{itemize}
\item $V^{\mathsf{InitFinConfs}}= \{(n-1,j)~|~0\leq j\leq n-1\}\cup \{(i,n-2)~|~0\leq i\leq n-1\}$, 
\item $V^{\mathsf{Connectors}}=\{(n-1,j)~|~n\leq j\leq c_t\}\cup \{(n,c_i)~|~1\leq i\leq t\}$, and 
\item $V^{\mathsf{Instances}}=\bigcup ^t_{i=1} V(G_i)$. 
\end{itemize}
\item $E(G)=\{((a,b),(c,d))~|~(a,b),(c,d)\in V(G), |a-c|+|b-d|=1$\}. 
\item $\mathsf{init}=((n-1,n-1),(n-1,n-2),(n-1,n-3),\ldots,(n-1,0))$.
\item $\mathsf{fin}=((0,n-2),(1,n-2),(2,n-2),\ldots,(n-1,n-2))$.
\end {enumerate}
   
The following observation about the reduction function $\mathsf{HamToSna}$ follows directly from its definition.

\begin{observation} \label{obs:polRed}
The reduction function $\mathsf{HamToSna}$ returns a valid instance of the {\sc Snake Game} problem on grid graphs. Moreover, it can be computed in polynomial time.
\end{observation}

\begin{figure}
    \includegraphics[page=2,width=\textwidth]{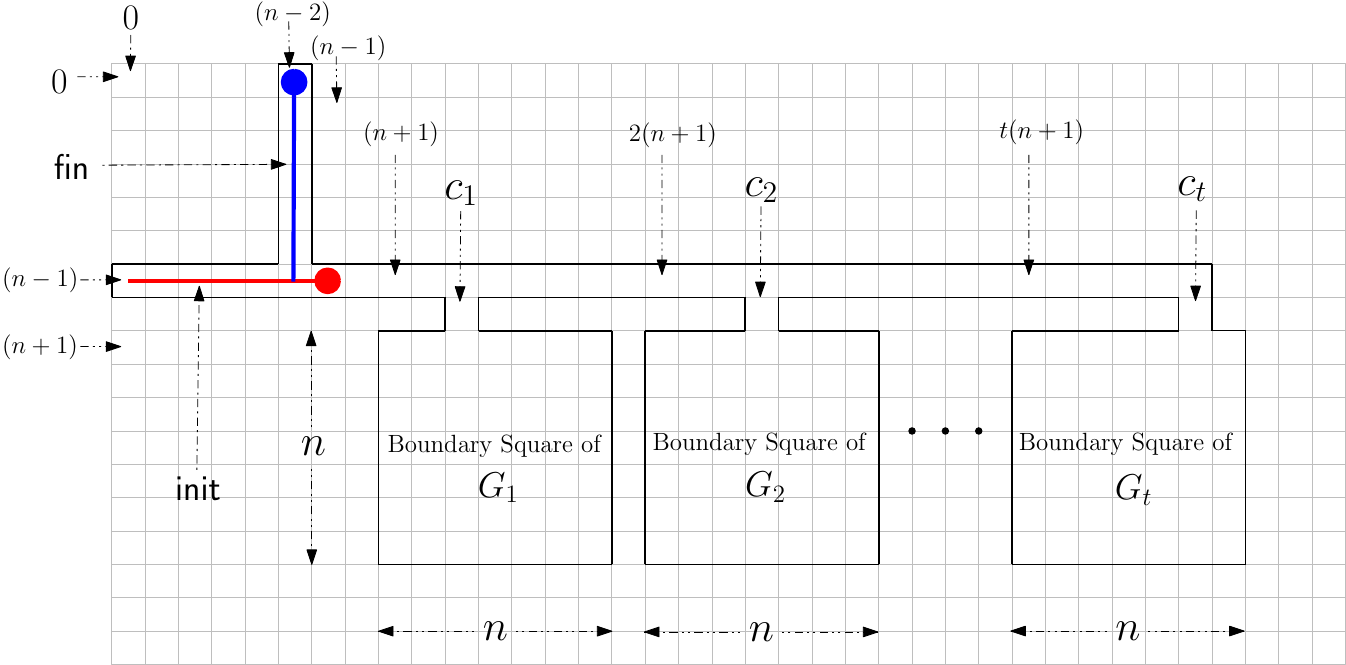}
    \caption{Graph constructed by the reduction function $\mathsf{HamToSna}$.}
  \label{fig:noKernel}
\end{figure}

\subsection{The Correctness of Construction}\label{sec:kernelReductionCorrectness}

To prove the correctness of our reduction, we need to show that the configuration $\mathsf{init}$ can reach the configuration $\mathsf{fin}$ if and only if there exists a \yes-instance among the $t$ instances of {\sc Hamiltonian Cycle}.
Towards this, we begin by showing that in order to reach the configuration $\mathsf{fin}$, the configuration $\mathsf{init}$ must first reach a configuration where the head of the snake is at one of the vertices of one of the $t$ instances of {\sc Hamiltonian Cycle}, that is, at $\bigcup_{i=1}^tV(G_i)$. Clearly, the snake must repeatedly make a move (i.e. a $1$-transition) until it has not reached $\mathsf{fin}$, and as long as it has also not reached a vertex in $\bigcup_{i=1}^tV(G_i)$, the choice of this move is very limited. Indeed, to see this, note that the graph we have constructed has no cycle that involves vertices that do not belong to $\bigcup_{i=1}^tV(G_i)$. Thus, starting at $\mathsf{init}$ and making one move at a time, the snake will never be able to ``turn around'' to reach $\mathsf{fin}$ unless it traverses a cycle in a graph~$G_i$.

To formalize this, we define the notion of the {\em head's neighbor set}. Roughly speaking, this set consists of the vertices to which the snake (in a given configuration) can move its head, thereby describing all possible $1$-transitions of the configuration. 

\begin{definition}[{\bf Head's Neighbor Set}]\label{def:headNe}
Let $\langle G,k,\mathsf{init},\mathsf{fin}\rangle$ be an instance of {\sc Snake Game}. Let $\mathsf{conf}=(v_1,\ldots ,v_k)$ be a configuration. The {\em head's neighbor set} of $\mathsf{conf}$, denote by $N^h(\mathsf{conf})$, is defined as $\{ u\in V(G)| \{ u,v_1\} \in E(G), u\neq v_i$ for every $1\leq i\leq k-1\}$.   
\end{definition}

The following observation follows directly from Definition~\ref{def:headNe}.

\begin{observation} \label{obs:1-trans}
Let $\langle G,k,\mathsf{init},\mathsf{fin}\rangle$ be an instance of {\sc Snake Game}. Let $\mathsf{conf}=(v_1,\ldots ,v_k)$ be a configuration. For every $u\in N^h(\mathsf{conf})$, let $\mathsf{conf}^u$ denotes the tuple $(u,v_1,\ldots,v_{k-1})$. Then, for every $u\in N^h(\mathsf{conf})$, $\mathsf{conf}^u$ is a configuration. Moreover, $\bigcup_{u\in N^h(\mathsf{conf})} \{ \mathsf{conf}^u\}$ is the set of all the configurations $\mathsf{conf'}$ such that ($\mathsf{conf},\mathsf{conf'}$) is a $1$-transition.  
\end{observation}

Now, we further observe that if the initial and final configurations are not equal, the final configuration is reachable from the initial configuration if and only if the final configuration is reachable from at least one of the configurations where the head of the snake moves to one of the vertices in the head's neighbor set.

\begin{observation} \label{lem:eqconf}
Let $\langle G,k,\mathsf{init},\mathsf{fin}\rangle$ be an instance of {\sc Snake Game} where $\mathsf{init}\neq \mathsf{fin}$. Then, $\langle G,k,\mathsf{init},\mathsf{fin}\rangle$ is a \yes-instance if and only if there exists a \yes-instance in $\{\langle G,k,\mathsf{init}^u,\mathsf{fin}\rangle~|~u\in N^h(\mathsf{init})\}$.   
\end{observation}

\begin{proof}
$(\Rightarrow)$ Suppose that $\langle G,k,\mathsf{init},\mathsf{fin}\rangle$ is a \yes-instance. Then, by Definitions~\ref{def:MoveToconf} and~\ref{def:SnakeGame}, the pair $(\mathsf{init}, \mathsf{fin})$ is an $\ell$-transition for some $\ell\in \mathbb{N}$. Moreover, by Definition~\ref{def:MoveToConfInEll}, there exists a tuple $(\mathsf{conf}_1 = \mathsf{init},\mathsf{conf}_2,\ldots,\mathsf{conf}_\ell = \mathsf{fin})$ such that, for every $1\leq i\leq \ell$, the pair $(\mathsf{conf}_i,\mathsf{conf}_{i+1})$ is a $1$-transition. As $\mathsf{init}\neq \mathsf{fin}$, we know that $\ell>1$. Moreover, as ($\mathsf{init},\mathsf{conf}_2$) is a $1$-transition, by Observation~\ref{obs:1-trans}, there exists  $u\in N^h(\mathsf{init})$ such that $\mathsf{init}^u=\mathsf{conf}_2$. So, we get that $(\mathsf{conf}_2=\mathsf{init}^u,\ldots,\mathsf{conf}_\ell=\mathsf{fin})$ is an $\ell-1$-transition. Thus, by Definitions~\ref{def:MoveToconf} and~\ref{def:SnakeGame}$, \langle G,k,\mathsf{init}^u,\mathsf{fin}\rangle$ is a \yes-instance.

$(\Leftarrow)$ Suppose that $\langle G,k,\mathsf{init}^u,\mathsf{fin}\rangle$ is a \yes-instance for some $u\in N^h(\mathsf{init})$. Then, by Definitions~\ref{def:MoveToconf} and~\ref{def:SnakeGame}, the pair $(\mathsf{init}^u, \mathsf{fin})$ is an $\ell$-transition for some $\ell\in \mathbb{N}$. Moreover, by Definition~\ref{def:MoveToConfInEll}, there exists a tuple $(\mathsf{conf}_1 = \mathsf{init}^u,\mathsf{conf}_2,\ldots,\mathsf{conf}_\ell = \mathsf{fin})$ such that, for every $1\leq i\leq \ell$, the pair $(\mathsf{conf}_i,\mathsf{conf}_{i+1})$ is a $1$-transition.  As $(\mathsf{init},\mathsf{init}^u)$ is a $1$-transition, we get that $(\mathsf{init},\mathsf{conf}_1,\mathsf{conf}_2,\ldots,\mathsf{conf}_\ell)$ is an $\ell+1$-transition. Thus, by Definitions~\ref{def:MoveToconf} and~\ref{def:SnakeGame}, $\langle G,k,\mathsf{init},\mathsf{fin}\rangle$ is a \yes-instance.	
\end{proof}       

We proceed as follows. For every $1\leq i \leq t$, let $\mathsf{init_i}$ denote the configuration whose head is $(r_{\mathrm{min}}^i, c_i)$ and the rest of its vertices does not belong to $V(G_i)$. (See Figure~\ref{fig:noKernelCons}). Then, we show that $\langle G,k,\mathsf{init},\mathsf{fin}\rangle$ is a \yes-instance if and only if at least one of the instances $\langle G,k,\mathsf{init_i},\mathsf{fin}\rangle$ is a \yes-instance.

\begin{figure}
    \includegraphics[page=2,width=\textwidth]{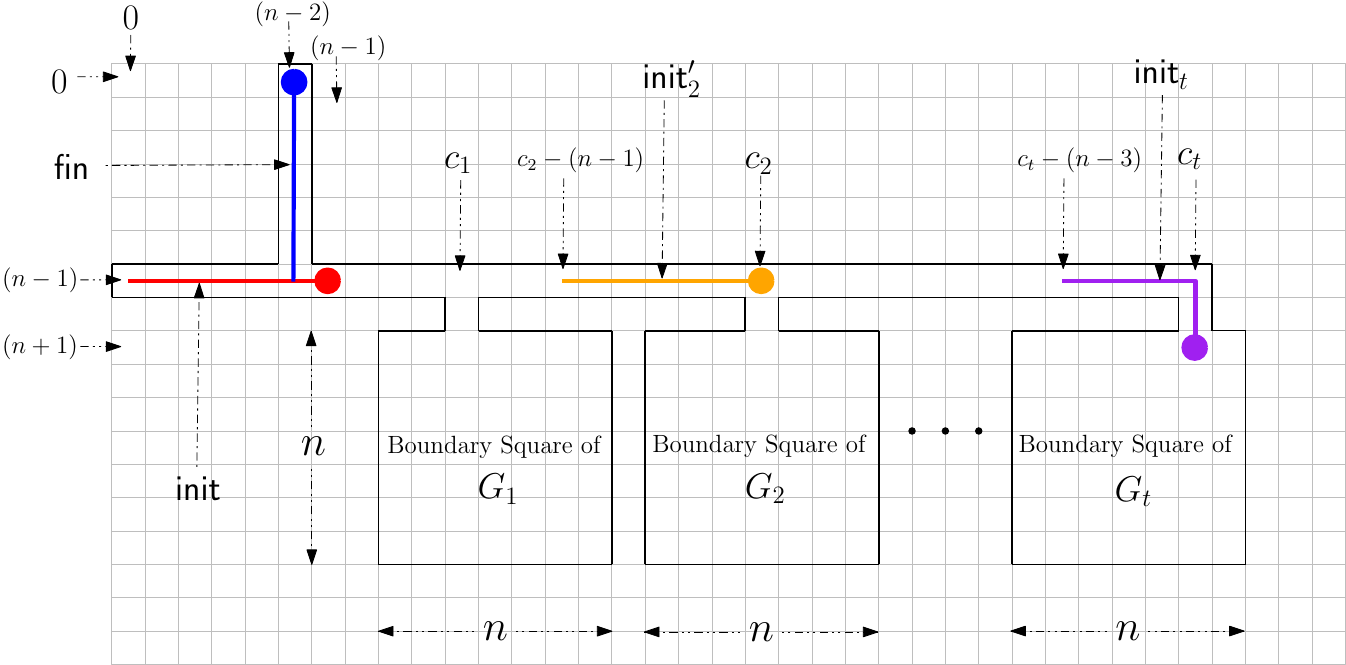}
    \caption{The configurations $\mathsf{init}$, $\mathsf{fin}$, $\mathsf{init}_i$ and $\mathsf{init}'_i$} for $1\leq i\leq t$ (specifically, $\mathsf{init}'_2$ and $\mathsf{init}_t$ are shown). 
  \label{fig:noKernelCons}
\end{figure}

\begin{lemma} \label{lem:noKer1}
Let $G_1,G_2,\ldots,G_t$ be $t$ instances of {\sc Hamiltonian Cycle} on grid graphs. Let $\mathsf{HamToSna}(G_1,G_2,\ldots,G_t)=\langle G,n,\mathsf{init},\mathsf{fin}\rangle$, and $\mathsf{init}_i:=((n+1,c_i),(n,c_i),(n-1,c_i),(n-1,c_i-1)\ldots (n-1,c_i-(n-3)))$ for every $1\leq i\leq t$ (see Figure~\ref{fig:noKernelCons}). Then, $\langle G,n,\mathsf{init},\mathsf{fin}\rangle$ is a \yes-instance if and only if there exists a \yes-instance in $\{\langle G,n,\mathsf{init}_i,\mathsf{fin}\rangle~|~1\leq i\leq t\}$.
\end{lemma}

\begin{proof}
For every $1\leq i\leq t$, denote $\mathsf{init}'_i=((n-1,c_i),(n-1,c_i-1),\ldots,(n-1,c_i-(n-1)))$. Refer to Figure~\ref{fig:noKernelCons}. Recall that $\mathsf{init}=((n-1,n-1),(n-1,n-2),(n-1,n-3),\ldots,(n-1,0))$, and notice that $N^h(\mathsf{init})=\{(n-1,n)\}$. So, by Observation \ref{lem:eqconf}, we know that $\langle G,n,\mathsf{init},\mathsf{fin}\rangle$ is a \yes-instance if and only if $\langle G,n,((n-1,n),(n-1,n-1),(n-1,n-2),\ldots,(n-1,1)),\mathsf{fin}\rangle$ is a \yes-instance. By iteratively repeating this step $c_1-n$ additional times, we derive that $\langle G,n,\mathsf{init},\mathsf{fin}\rangle$ is a \yes-instance if and only if $\langle G,n,\mathsf{init}'_1,\mathsf{fin}\rangle$ is a \yes-instance.

Now, for any $1\leq i\leq t-1$, note that $N^h(\mathsf{init}'_i)=\{(n,c_i),(n-1,c_i+1)\}$ and $N^h(((n,c_i),(n-1,c_i),\ldots,(n-1,c_i-(n-2))))=\{(n+1,c_i)\}$. So, by Observation \ref{lem:eqconf}, we know that $\langle G,n,\mathsf{init}'_i,\mathsf{fin}\rangle$ is a \yes-instance if and only if at least one of the instances $\langle G,n,\mathsf{init}_i,\mathsf{fin}\rangle$ and $\langle G,n,((n-1,c_i+1),(n-1,c_i),\ldots,(n-1,c_i-(n-2))),\mathsf{fin}\rangle \}$ is a \yes-instance. Furthermore, repeating the argument of the first paragraph for $((n-1,c_i+1),(n-1,c_i),\ldots,(n-1,c_i-(n-2)))$, we obtain that $\langle G,n,((n-1,c_i+1),(n-1,c_i),\ldots,(n-1,c_i-(n-2))),\mathsf{fin}\rangle$ is a \yes-instance if and only if $\langle G,n,\mathsf{init}'_{i+1},\mathsf{fin}\rangle$ is a \yes-instance.

Putting the two claims above together, we derive that $\langle G,n,\mathsf{init},\mathsf{fin}\rangle$ is a \yes-instance if and only if there exists a \yes-instance in $\{\langle G,n,\mathsf{init}_i,\mathsf{fin}\rangle~|~1\leq i\leq t-1\}\cup\{ \langle G,n,\mathsf{init}'_{t},\mathsf{fin}\rangle\}$. As $N^h(\mathsf{init}'_t)=\{(n,c_t)\}$ and $N^h(((n,c_t),(n-1,c_t),\ldots,(n-1,c_t-(n-2))))=\{(n+1,c_t)\}$, by Observation \ref{lem:eqconf}, we know that $\langle G,n,\mathsf{init}'_{t},\mathsf{fin}\rangle$ is a \yes-instance if and only if $\langle G,n,\mathsf{init}_{t},\mathsf{fin}\rangle$ is a \yes-instance. In turn, this completes the proof.             
\end{proof}

Up until now, we have shown that the snake must reach one of the graphs corresponding to the $t$ instances of {\sc Hamiltonian Cycle} in order to reach $\mathsf{fin}$.
Next, we show that once the snake reaches one of these instances, it can reach $\mathsf{fin}$ if and only if that instance is a \yes-instance of {\sc Hamiltonian Cycle}. For the sake of clarity, we split the proof into two lemmas as follows.
     
\begin{lemma} \label{lem:noKer2}
If $G_i$ is a \yes-instance of the {\sc Hamiltonian Cycle} problem for some $1 \leq i\leq t$, then $\langle G,n,\mathsf{init}_i,\mathsf{fin}\rangle$ is a \yes-instance of {\sc Snake Game}. 
\end {lemma}

\begin{proof}
Assume that $G_i$ is a \yes-instance of the {\sc Hamiltonian Cycle} problem for some $1 \leq i\leq t$. Let $V(G_i) = \{a_1, a_2, \ldots, a_n\}$ where $a_1 = (n+1,c_i)$. So, there exists a simple cycle $C=a_1-\ldots-a_n-a_1$ in $G_i$.
Note that $V(\mathsf{init}_i) \cap V(G_i) = \{a_1\}$, $\{a_1,a_2\}\in E(G_i)$ and $a_1\neq a_2$. 
Therefore, $(\mathsf{init}_i=(a_1,(n,c_i),(n-1,c_i),(n-1,c_i-1),\ldots (n-1,c_i-(n-3))),(a_2,a_1,(n,c_i),(n-1,c_i),(n-1,c_i-1),\ldots (n-1,c_i-(n-4))))$ is a $1$-transition.

By iteratively repeating the above transformation $n-1$ additional times, we get that the pair 
$(\mathsf{init}_i=(a_1,(n,c_i),\ldots (n-1,c_i-(n-3))), (a_1,a_n,\ldots a_2))$ is an $n$-transition. So, by Definition~\ref{def:MoveToconf}, we know that $\mathsf{init}_i$ can reach $(a_1,a_n,\ldots a_2)$. 
Moreover, by Observation \ref{lem:eqconf} (as in the proof of Lemma \ref{lem:noKer1}), it is easily seen that $(a_1,a_{n},\ldots a_2)$ can reach $\mathsf{fin}$. Combining these two statements together, we get that $\mathsf{init}_i$ can reach $\mathsf{fin}$. By Definition~\ref{def:SnakeGame}, we thus conclude that $\langle G,n,\mathsf{init}_i,\mathsf{fin}\rangle$ is a \yes-instance of the {\sc Snake Game} problem.              
\end{proof}

Now, we show that the opposite direction of Lemma \ref{lem:noKer2} is also true. 
Intuitively, if the snake is located at $\mathsf{init}_i$, then it must traverse a Hamiltonian cycle in $G_i$ in order to ``exit'' $G_i$ and reach $\mathsf{fin}$---in particular, the Hamiltonicity condition has to be satisfied otherwise the snake will intersect itself while trying to ``exit'' $G_i$.  

\begin{lemma}\label{lem:noKer3}
If $\langle G,n,\mathsf{init}_i,\mathsf{fin}\rangle$ is a \yes-instance of {\sc Snake Game} for some $1\leq i\leq t$, then $G_i$ is a \yes-instance of the {\sc Hamiltonian Cycle} problem.
\end {lemma}

\begin{proof}
Assume that $\langle G,n,\mathsf{init}_i,\mathsf{fin}\rangle$ is a \yes-instance of the {\sc Snake Game} problem for some $1\leq i\leq t$. By Definitions~\ref{def:MoveToconf} and~\ref{def:SnakeGame}, the pair $(\mathsf{init}_i, \mathsf{fin})$ is an $\ell$-transition for some $\ell\in \mathbb{N}$. Moreover, by Definition~\ref{def:MoveToConfInEll}, there exists a tuple $T = (\mathsf{conf}_1 = \mathsf{init}_i,\mathsf{conf}_2,\ldots,\mathsf{conf}_{\ell+1} = \mathsf{fin})$ such that, for every $1\leq i\leq \ell$, the pair $(\mathsf{conf}_i,\mathsf{conf}_{i+1})$ is a $1$-transition. For every $1\leq j\leq \ell +1$, let $\mathsf{conf}_j=(v^j_1,\ldots,v^j_n)$.

First, we show that $\ell +1>n$. By way of contradiction, suppose that $\ell +1\leq n$. Note that $|V(\mathsf{init}_i)| = |V(\mathsf{fin})| = n$. As $(\mathsf{init}_i, \mathsf{fin})$ is an $\ell$-transition for some $\ell \leq n-1$, by Lemma~\ref{lem:confGraph}, $|V(\mathsf{init}_i) \cap V(\mathsf{fin})| \geq n-\ell \geq 1$. By construction, $|V(\mathsf{init}_i) \cap V(\mathsf{fin})| \leq 1$, so we get that $|V(\mathsf{init}_i) \cap V(\mathsf{fin})| = 1$, which in turn gives us that $\ell = n-1$. So, by Corollary~\ref{cor:(k-1)Transition}, head of the configuration $\mathsf{init}_i$ should be same as the tail of the configuration $\mathsf{fin}$, which is a contradiction. Thus, we get that $\ell +1>n$, hence it is well defined to consider the first $n+1$ configurations of the tuple $T$. 

Let $a = (n+1,c_i)$. Observe that $v^1_1 = a \in V(G_i)$. Furthermore, observe that $\mathsf{HamToSna}$ constructs $G$ in a way that $a$ is the only vertex in $V(G_i)$ that is connected to vertices not in $V(G_i)$. As $v_2^1=(n,c_i)$, by the definition of $1$-transition, $v_1^2\neq (n,c_i)$, $v_1^2\neq v^1_1$ and $\{ v_1^1, v_1^2\}  \in E(G)$. Moreover, $(n,c_i)$ is the only neighbor of $v^1_1$ which is not in $V(G_i)$, so $v_1^2 \in V(G_i)$ and $(v_1^1, v_1^2) \in E(G_i)$. By repeating the same argument for the first $n$ configurations of the tuple $T$, we get that for every $1 \leq k < j \leq n$, $v_1^j \in V(G_i)$, $v_1^k \neq v_1^j$ and for every $1 \leq j < n$, $\{ v_1^j, v_1^{j+1}\} \in E(G_i)$. In turn, we get that for every $1 \leq k < j \leq n$, $v_j^n \in V(G_i)$, $v_k^n \neq v_j^n$ and for every $1 \leq j < n$, $(v_j^n, v_{j+1}^n) \in E(G_i)$. Therefore, $v^n_1-\cdots-v^n_n$ is a Hamiltonian {\em path} in $G_i$.

Towards the proof that $v^n_1$ is adjacent to $v^n_n$ in $G_i$, notice that $v^n_n = a \neq v^n_1$. Moreover, from tuple $T$, we know that $\mathsf{fin}$ is reachable from $\mathsf{conf}_n$. As $a$ is the only vertex in $V(G_i)$ that is connected to vertices not in $V(G_i)$ and by the definition of $1$-transition, we get that $v_1^{n+1} = a$. This implies that $\{ a, v^1_n\} \in E(G_i)$. Therefore, we get that $a=v^n_1-\ldots,v^n_n-a$ is a Hamiltonian cycle in $G_i$. We thus conclude that $G_i$ is a \yes-instance of the {\sc Hamiltonian Cycle} problem.                
\end{proof}

We now combine Lemmas~\ref{lem:noKer2} and~\ref{lem:noKer3} to conclude the correctness of the construction. 

\begin{lemma}\label{lem:algoWorks}
Let $G_1,G_2,\ldots,G_t$ be $t$ instances of {\sc Hamiltonian Cycle} on grid graphs. Then, at least one of the instances $G_1,G_2,\ldots,G_t$ is a \yes-instance of {\sc Hamiltonian Cycle} if and only if $\mathsf{HamToSna}(G_1,G_2,\ldots,G_t)$ is a \yes-instance of {\sc Snake Game}.
\end{lemma}

\begin{proof}
$(\Rightarrow)$ Assume that at least one of $G_1,G_2,\ldots,G_t$ is a \yes-instance of {\sc Hamiltonian Cycle}. Let $i$ be such that $G_i$ is a \yes-instance of {\sc Hamiltonian Cycle}.  From Lemma~\ref{lem:noKer2}, we get that $\langle G,n,\mathsf{init}_i,\mathsf{fin}\rangle$ is a \yes-instance of {\sc Snake Game}. Therefore, at least one instance in $\{\langle G,n,\mathsf{init}_i,\mathsf{fin}\rangle~|~1\leq i\leq t\}$ is a \yes-instance. So, from Lemma~\ref{lem:noKer1}, we get that $\mathsf{HamToSna}$ $(G_1,G_2,\ldots,G_t)=\langle G,n,\mathsf{init},\mathsf{fin}\rangle$ is a \yes-instance of {\sc Snake Game}.

$(\Leftarrow)$ Assume that $\mathsf{HamToSna}(G_1,G_2,\ldots,G_t)=\langle G,n,\mathsf{init},\mathsf{fin}\rangle$ is a \yes-instance of {\sc Snake Game}. Then, from Lemma~\ref{lem:noKer1}, we get that at least one instance in $\{\langle G,n,\mathsf{init}_i,$ $\mathsf{fin}\rangle~|~1\leq i\leq t\}$ is a \yes-instance. Let $i$ be such that $\langle G,n,\mathsf{init}_i,\mathsf{fin}\rangle$ is a \yes-instance. In turn, from Lemma~\ref{lem:noKer3}, we get that $G_i$ is a \yes-instance of {\sc Hamiltonian Cycle}. So, at least one of $G_1,G_2,\ldots,G_t$ is a \yes-instance of {\sc Hamiltonian Cycle}.      
\end{proof}

\subsection{Conclusion of the Proof}

Up until now, we have proved the correctness of our reduction. Now, we are ready to prove the result of the (unlikely) existence of a polynomial kernel (or even compression) for {\sc Snake Game} on grid graphs.   

\begin{lemma} \label{lem:lastlemma}
The {\sc Hamiltonian Cycle} problem on grid graphs cross-composes into the {\sc Snake Game} problem on grid graphs.
\end{lemma} 

\begin{proof}
By Observation \ref{obs:polRed}, $\mathsf{HamToSna}$ is computable in polynomial-time. By Lemma \ref{lem:algoWorks}, given any $t$ instances of the {\sc Hamiltonian Cycle} problem on grid graphs, $\mathsf{HamToSna}$ evaluates to a \yes-instance of the {\sc Snake Game} problem on grid graphs if and only if at least one of the $t$ instances of {\sc Hamiltonian Cycle} is a \yes-instance. Thus, by Definition \ref{def:cross-comp}, {\sc Hamiltonian Cycle} on grid graphs cross-composes into {\sc Snake Game} on grid graphs.
\end {proof}

Thus, from Lemma~\ref{lem:lastlemma} and Proposition~\ref{prop:noKern}, we get the following theorem.

\begin{theorem}
The {\sc Snake Game} problem on grid graphs does not admit a polynomial compression unless \NP $\subseteq$\coNPpoly.
\end{theorem}

\begin{proof}
By Lemma \ref{lem:lastlemma}, {\sc Hamiltonian Cycle} on grid graphs cross-composes into {\sc Snake Game} on grid graphs. As {\sc Hamiltonian Cycle} on grid graphs is an \NP-hard problem~\cite{DBLP:journals/jal/PapadimitriouV84}, Proposition \ref{prop:noKern} implies that {\sc Snake Game} on grid graphs does not admit a polynomial compression unless \NP$\subseteq$\coNPpoly.
\end{proof}
\section{Treewidth Reduction on General Graphs}\label{sec:treewidth}

In this section, we prove that given an instance of {\sc Snake Game}, we can get an equivalent instance of {\sc Snake Game} where the treewidth of the graph is bounded by a polynomial in the size of the snake. Our proof is based on a theorem that, for a given undirected graph $G$ and $k\in \mathbb{N}$, asserts the following statement. We can, in time $\OO(|V(G)| + |E(G)|)$, either find a specific ``pattern'' in the graph called a wall, or determine that $G$ has treewidth $k^{\OO(1)}$. Given an instance $\langle G,k,\mathsf{init},\mathsf{fin}\rangle$ of {\sc Snake Game}, we show how we use that pattern in $G$ in order to decrease the number of vertices in $G$, yet retain an equivalent instance of {\sc Snake Game}. Specifically, we utilize rerouting arguments to find a so called {\em irrelevant edge} (with respect to contraction), inspired by the classic work of Robertson and Seymour \cite{article:Thedisjointpathsproblem}.  Because either way we obtain an exploitable structure (a wall or small treewidth), this method is known as ``win/win approach'' \cite{DBLP:books/sp/CyganFKLMPPS15}. By applying this method repeatedly, we derive in polynomial time an equivalent instance of {\sc Snake Game}, $\langle G',k,\mathsf{init},\mathsf{fin}\rangle$, such that $G'$ has treewidth $k^{\OO(1)}$.

First, we present the formal definition of the specific pattern we are looking for (see Figure \ref{fig:4-wall}).      

\begin{definition}[{\bf Elementary $r$-Wall}]\label{def:ElemWall}
Let $r\in \mathbb{N}$. Let $G_r$ be the $r\times 2r$-grid, i.e. $V(G_r)=\{(i,j)~|~i\in \{1,\ldots,r\}, j\in \{1,\ldots,2r\} \}$ and $E(G_r)=\{ \{(i,j),(i',j')\} ~|~|i-i'|+|j-j'|=1\}$. The {\em elementary $r$-wall} is the graph obtained from $G_r$ by deleting all edges $\{ (2i-1, 2j-1),(2i,2j-1)\}$
for $i\in \{ 1,2,\ldots,\lfloor r/2\rfloor \}$ and $j\in \{ 1,2,\ldots,r\}$ and all edges $\{ (2i,2j),(2i+1,2j)\}$
for $i\in \{ 1,2,\ldots,\lfloor(r-1)/2\rfloor \}$ and $j\in \{ 1,2,\ldots,r\}$ , and then deleting the two resulting vertices
of degree $1$. 
\end{definition}

\begin{figure}
\centering
    \includegraphics[width=0.75\textwidth]{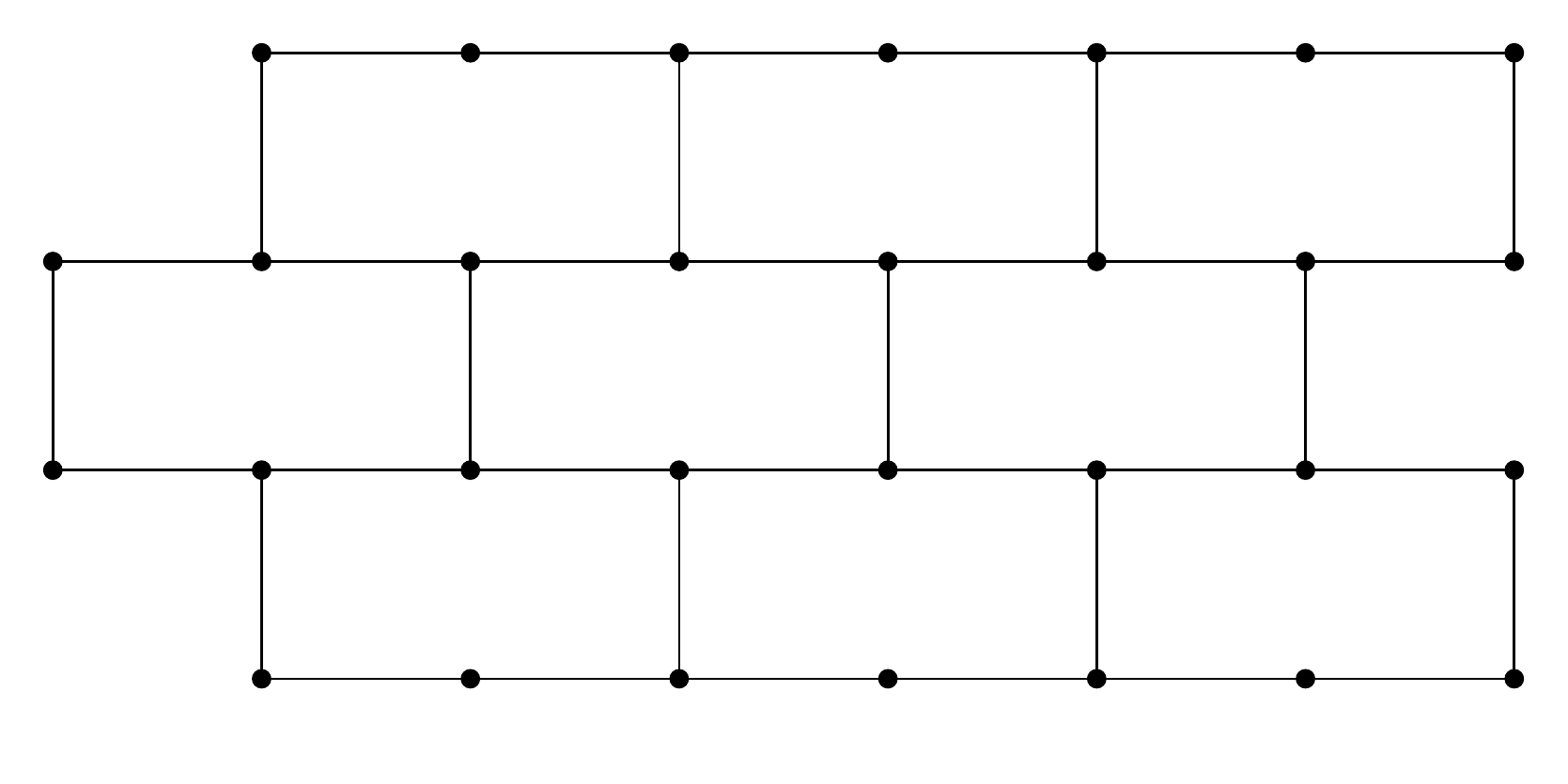}
    \caption{An elementary 4-wall.}
  \label{fig:4-wall}
\end{figure}

\begin{definition}[{\bf $r$-Wall}]\label{def:Wall}
Let $r\in \mathbb{N}$. A graph $G$ is an {\em $r$-wall} if it can be obtained from an elementary $r$-wall by subdividing edges.
\end {definition}

We would like to exploit the notion of an $r$-wall to present our reduction. To this end, we say that a graph $G$ {\em contains an $r$-wall} if there exists a subgraph $H$ of $G$ such that $H$ is an $r$-wall. 
Now, we aim to show that, given an instance $\langle G,k,\mathsf{init},\mathsf{fin}\rangle$ of {\sc Snake Game} where $G$ contains a $3\sqrt{k}$-wall, we can efficiently decrease the number of vertices in $G$ and still retain an equivalent instance of {\sc Snake Game}. In particular, we consider a $3\sqrt{k}$-wall because of its following special property. For any two vertices $s$ and $t$ in a $3\sqrt{k}$-wall $H$, there exists a simple path (or simple cycle if $s=t$) in $H$ that starts with $s$ and ends with $t$ whose size is at least $k$. Moreover, as it follows from the next observation, this property is preserved even if we contract any edge in $H$. 

\begin{observation}\label{obs:kpathexis}
Let $H$ be a $3\sqrt{k}$-wall. Let $e=\{ u,v\} \in E(H)$. Then, for every pair of vertices $s,t\in V(H/e)$, there exists a (simple) path $P$ (or cycle if $s=t$) in $H/e$ between $s$ and $t$ such that the size of $P$ is at least $k$. 
\end{observation}
 
Next, we consider the case where we have two configurations $\mathsf{conf}_1$ and $\mathsf{conf}_2$, and a (simple) path of size at least $k$ between the head of $\mathsf{conf}_1$ and the tail of $\mathsf{conf}_2$. In addition, we suppose that the path does not ``intersect'' $\mathsf{conf}_1$ and $\mathsf{conf}_2$ besides at these two vertices.  We prove (in Lemma \ref{lem:canREach}) that, in this case, $\mathsf{conf}_2$ must be reachable from $\mathsf{conf}_1$. Intuitively, the snake can move from $\mathsf{conf}_1$ through the path and thereby it can reach $\mathsf{conf}_2$. 
To this end, we start with a simple observation.  
In this observation, we consider the case where we have a configuration $\mathsf{conf}_1$ and a (simple) path $P$ of size at least $k$ that starts at the head of $\mathsf{conf}_1$, ends at a vertex $t$, and the remainder of $P$ and $\mathsf{conf}_1$ do not intersect. Informally, our observation states that in this case, by making the snake traverse $P$, we get that $\mathsf{conf}_1$ can reach some configuration $\mathsf{conf}_2$ where $t$ is the head, and the rest of the vertices of the snake is located on $P$.
Moreover, if $P$ is of size exactly $k$, then it corresponds to a configuration that is reachable from $\mathsf{conf}_1$. Notice that in this case, the path is reachable from $\mathsf{conf}_1$ in a ``reversed'' order of vertices, as the first vertex of the path ``becomes'' the tail of the snake when the snake traverses the path.            

\begin{observation}\label{obs:canReach22}
Let $\langle G,k,\mathsf{init},\mathsf{fin}\rangle$ be an instance of {\sc Snake Game}. Let $\mathsf{conf}_1=(u_1,\ldots,u_k)$ be a configuration. Let $P=(v_1,\ldots,v_\ell)$ be a simple path or cycle in $G$ of size at least $k$, such that $u_1=v_1$, and for all $2\leq i\leq k$ and $2\leq j\leq \ell$, $u_i\neq v_j$.  
Then, there exists a configuration $\mathsf{conf}_2=(u'_1,\ldots,u'_k)$ such that $u'_1=v_\ell$, for all $1\leq i\leq k$, $u'_i\in V(P)$, and $\mathsf{conf}_1$ can reach $\mathsf{conf}_2$. 

In addition, if $P$ is a path of size $k$, then $\mathsf{conf}=(v_\ell,\ldots,v_1)$ is a configuration reachable from $\mathsf{conf}_1$. 
\end{observation}

Now, we use Observation \ref{obs:canReach22} to prove the following lemma.
	
\begin{figure}
\centering
    \includegraphics[width=0.75\textwidth]{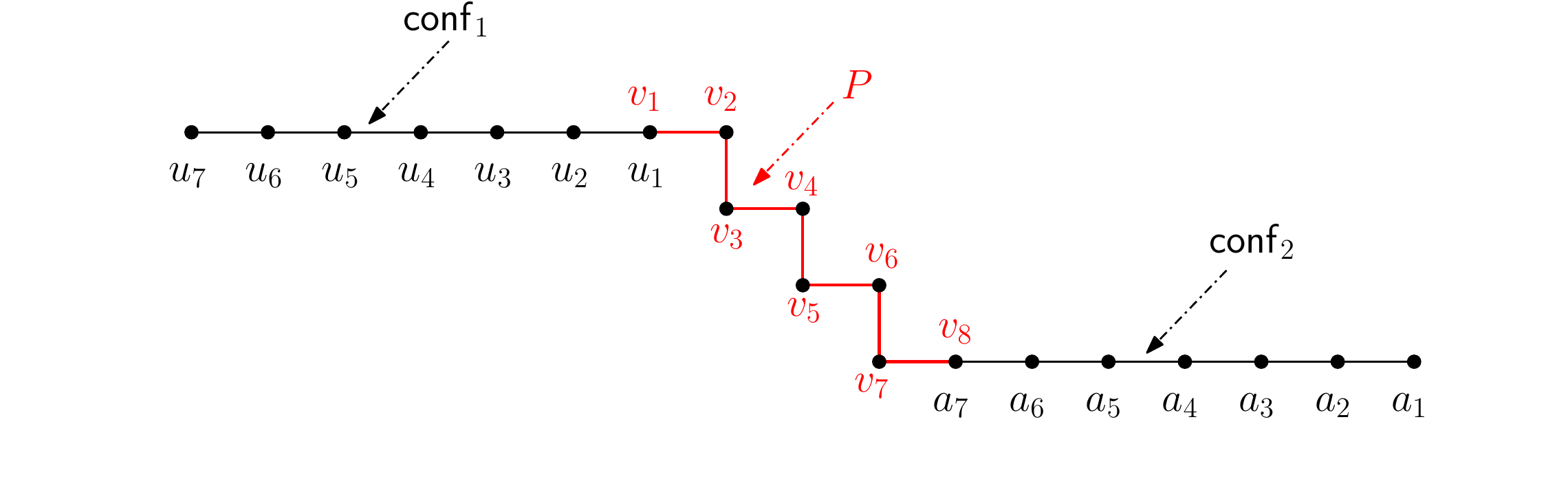}
    \caption{An example of the pre-conditions of Lemma \ref{lem:canREach} for $k=7$ and $\ell=8$.}
  \label{fig:lemma5position}
\end{figure}	
	
\begin{lemma} \label{lem:canREach}
Let $\langle G,k,\mathsf{init},\mathsf{fin}\rangle$ be an instance of {\sc Snake Game}. Let $\mathsf{conf}_1=(u_1,\ldots,u_k)$ and $\mathsf{conf}_2=(a_1,\ldots,a_k)$  be two configurations. Let $P=(v_1,\ldots,v_\ell)$ be a simple path or cycle in $G$ of size at least $k$ such that the following conditions hold (see Figure \ref{fig:lemma5position}).
\begin{itemize}
\item $v_1=u_1$ and $v_\ell=a_k$.
\item For every $2\leq i\leq k$ and $2\leq j\leq \ell$, $u_i\neq v_j$.
\item For every $1\leq i\leq k-1$ and $1\leq j\leq \ell-1$, $a_i\neq v_j$.
\end{itemize}
Then, $\mathsf{conf}_1$ can reach $\mathsf{conf}_2$.
\end{lemma}

\begin{proof}
Notice that $\mathsf{conf}_1=(u_1,\ldots,u_k)$ and $P=(v_1,\ldots,v_\ell)$ satisfy the conditions of Observation \ref{obs:canReach22}. Therefore, by this observation, there exists a configuration $\mathsf{conf}'_2=(u'_1,\ldots,u'_k)$ such that the following conditions are satisfied:
\begin{itemize}
\item $u'_1=v_\ell$.
\item For every $1\leq i\leq k$, $u'_i\in V(P)$.
\item $\mathsf{conf}_1$ can reach $\mathsf{conf}'_2$.  
\end{itemize} 
Now, note that $\mathsf{conf}_2=(a_1,\ldots,a_k)$ corresponds to a simple path in $G$, which can be denoted by $P'=(a'_1=a_k,\ldots,a'_k=a_1)$. In particular, the following conditions are satisfied:
\begin{itemize}
\item $a'_1=u'_1$.
\item For every $2\leq i\leq k$ and $2\leq j\leq k$, $u'_i\neq a'_j$.
\item $P'$ is a simple path in $G$ of size $k$. 
\end{itemize}  
Therefore, by Observation \ref{obs:canReach22}, $\mathsf{conf}'_2$ can reach $\mathsf{conf}_2$.
In turn, we conclude that $\mathsf{conf}_1$ can reach $\mathsf{conf}_2$.   
\end{proof}

Now, given an instance $\langle G,k,\mathsf{init},\mathsf{fin}\rangle$ of {\sc Snake Game}, we aim to show that if $G$ contains a $3\sqrt{k}$-wall $H$ that intersects neither $\mathsf{init}$ nor $\mathsf{fin}$, then for any $e\in E(H)$, $\langle G/e,k,\mathsf{init},\mathsf{fin}\rangle$ is an instance equivalent to $\langle G,k,\mathsf{init},\mathsf{fin}\rangle$.
The next lemma states the ``harder'' direction of this claim. Intuitively, if $\langle G,k,\mathsf{init},\mathsf{fin}\rangle$ is a \yes-instance, then there exists an $\ell$-transition from $\mathsf{init}$ to $\mathsf{fin}$ for some $\ell \in \mathbb{N}$. So, we can consider the first and last configurations that intersect $V(H)$ (if such configurations exist), say $\mathsf{conf}_1$ and $\mathsf{conf}_2$. Because $\mathsf{conf}_1$ is the first configuration that intersects $V(H)$, we show (in the proof of the next lemma) that the head of the snake in $\mathsf{conf}_1$ is a vertex from $V(H)$, and the rest of the vertices of the snake do not intersect $V(H)$. Similarly, we show that the tail of the snake in $\mathsf{conf}_2$ is a vertex from $V(H)$, and the rest of the vertices of the snake do not intersect $V(H)$. By Observation \ref{obs:kpathexis}, we get that there is a path in $H$ of size at least $k$ between the head of $\mathsf{conf}_1$ and the tail of $\mathsf{conf}_2$. By traversing this path from $\mathsf{conf}_1$, we show that snake can reach $\mathsf{conf}_2$. Therefore, we get that (in $\langle G/e,k,\mathsf{init},\mathsf{fin}\rangle$) $\mathsf{init}$ can reach $\mathsf{conf}_1$, $\mathsf{conf}_1$ can reach $\mathsf{conf}_2$, and $\mathsf{conf}_2$ can reach $\mathsf{fin}$. So, in turn we conclude that $\langle G/e,k,\mathsf{init},\mathsf{fin}\rangle$ is a \yes-instance.             

\begin{lemma} \label{lem:proof1}
Let $\langle G,k,\mathsf{init},\mathsf{fin}\rangle$ be an instance of {\sc Snake Game}, where $G$ contains a $3\sqrt{k}$-wall $H$ having vertices from neither $\mathsf{init}$ nor $\mathsf{fin}$. Let $e=\{u,v\} \in E(H)$. If $\langle G,k,\mathsf{init},\mathsf{fin}\rangle$ is a \yes-instance, then $\langle G/e,k,\mathsf{init},\mathsf{fin}\rangle$ is a \yes-instance.    
\end{lemma}

\begin{proof}
Suppose that $\langle G,k,\mathsf{init},\mathsf{fin}\rangle$ is a \yes-instance. Then, there exists an $\ell$-transition $T=(\mathsf{conf}_1,\ldots,\mathsf{conf}_{\ell+1})$ such that $\mathsf{conf}_1=\mathsf{init}$ and $\mathsf{conf}_{\ell+1}=\mathsf{fin}$.

 First, assume that for every $1\leq i\leq \ell +1$, $\mathsf{conf}_i$ does not contain any vertices from $V(H)$. Then, $(\mathsf{conf}_1,\ldots,\mathsf{conf}_{\ell+1})$ is also an $\ell$-transition in $\langle G/e,k,\mathsf{init},\mathsf{fin}\rangle$, so it is a \yes-instance.
Now, assume that there exists $1\leq i\leq \ell +1$ such that $\mathsf{conf}_i$ contains at least one vertex from $V(H)$. Let $i$ and $j$ be such that $\mathsf{conf}_i$ is the first configuration and $\mathsf{conf}_j$ is the last configuration in $T$ that contains a vertex from $V(H)$. Notice that $1<i\leq j<\ell+1$ because we assumed that $H$ contains vertices from neither $\mathsf{init}$ nor $\mathsf{fin}$. Denote $\mathsf{conf}_i=(a_1,\ldots,a_k)$, and observe that $a_1\in V(H)$ and for every $2\leq m\leq k$, $a_m\notin V(H)$ (otherwise $\mathsf{conf}_{i}$ would not be the first configuration to contain a vertex from $H$, a contradiction). Similarly, denote $\mathsf{conf}_j=(b_1,\ldots,b_k)$, and observe that $b_k\in V(H)$ and for every $1\leq m\leq k-1$, $b_m\notin V(H)$.

Recall that $v_e\in V(G/e)$ denotes the new vertex that was created by the contraction of $e$. Now, let $\mathsf{conf}'_i$ be a configuration and $a'_1$ a vertex defined as follows. If $a_1=u$ or $a_1=v$ then $\mathsf{conf}'_i=(v_e,a_2,\ldots,a_k)$ and $a'_1=v_e$, else $\mathsf{conf}'_i=\mathsf{conf}_i$ and $a'_1=a_1$. Similarly, if $b_k=u$ or $b_k=v$ then let $\mathsf{conf}'_j=(b_1,\ldots,b_{k-1},v_e)$ and $b'_k=v_e$, otherwise $\mathsf{conf}'_j=\mathsf{conf}_j$ and $b'_k=b_k$. Now, notice that $(\mathsf{conf}_1,\ldots,\mathsf{conf}_{i-1},\mathsf{conf}'_{i})$ is an $(i-1)$-transition and $(\mathsf{conf}'_{j},\mathsf{conf}_{j+1},\ldots,\mathsf{conf}_{\ell+1})$ is an $(\ell+1-j)$-transition in $\langle G/e,k,\mathsf{init},\mathsf{fin}\rangle$. Therefore, we know that both $\mathsf{conf}_1$ can reach $\mathsf{conf}'_{i}$ and $\mathsf{conf}'_j$ can reach $\mathsf{conf}_{\ell+1}$ in $\langle G/e,k,\mathsf{init},\mathsf{fin}\rangle$.  
  
By Observation \ref{obs:kpathexis}, we know that there exists a simple path $P=(u_1,\ldots,u_\ell$) (or cycle if $a'_1=b'_k$) of size at least $k$ such that $u_1=a'_1$, $u_\ell=b'_k$, and for every $1\leq i\leq \ell$, $u_i\in V(H/e)$. 
Therefore, by Lemma \ref{lem:canREach} we get that $\mathsf{conf}'_{i}$ can reach $\mathsf{conf}'_{j}$ in $\langle G/e,k,\mathsf{init},\mathsf{fin}\rangle$.
Overall, as $\mathsf{conf}_1=\mathsf{init}$ can reach $\mathsf{conf}'_{i}$, $\mathsf{conf}'_{i}$ can reach $\mathsf{conf}'_{j}$ and $\mathsf{conf}'_j$ can reach $\mathsf{conf}_{\ell+1}=\mathsf{fin}$, we conclude that $\mathsf{init}$ can reach $\mathsf{fin}$ in $\langle G/e,k,\mathsf{init},\mathsf{fin}\rangle$. Therefore $\langle G/e,k,\mathsf{init},\mathsf{fin}\rangle$ \yes-instance.                  
\end{proof}

Now, we consider the opposite direction of Lemma \ref{lem:proof1}. Roughly speaking, the idea of the proof is as follows. If $\langle G/e,k,\mathsf{init},\mathsf{fin}\rangle$ is a \yes-instance, then we have an $\ell$-transition $(\mathsf{conf}_1,\ldots,\mathsf{conf}_{\ell+1})$ such that $\mathsf{conf}_1=\mathsf{init}$ and $\mathsf{conf}_{\ell+1}=\mathsf{fin}$. We aim to ``fix'' $(\mathsf{conf}_1,\ldots,\mathsf{conf}_{\ell+1})$ to ``fit'' $\langle G,k,\mathsf{init},\mathsf{fin}\rangle$ by substituting all appearances of $v_e$. Intuitively, if the snake traverse a path that passes through $v_e$ in $G/e$, it can visit at least one among $u$ and $v$ instead in $G$.         

\begin{lemma} \label{lem:proof2}
Let $\langle G,k,\mathsf{init},\mathsf{fin}\rangle$ be an instance of {\sc Snake Game}, where $G$ contains a $3\sqrt{k}$-wall $H$ having vertices from neither $\mathsf{init}$ nor $\mathsf{fin}$. Let $e=\{u,v\} \in E(H)$. If $\langle G/e,k,\mathsf{init},\mathsf{fin}\rangle$ is a \yes-instance, then $\langle G,k,\mathsf{init},\mathsf{fin}\rangle$ is a \yes-instance.     
\end{lemma}

\begin{proof}
Suppose that $\langle G/e,k,\mathsf{init},\mathsf{fin}\rangle$ is a \yes-instance. Then, there exists an $\ell$-transition $(\mathsf{conf}_1$ $,\ldots, \mathsf{conf}_{\ell+1})$ such that $\mathsf{conf}_1=\mathsf{init}$ and $\mathsf{conf}_{\ell+1}=\mathsf{fin}$. If for every $1\leq i\leq \ell +1$, $v_e\notin \mathsf{conf}_i$, then $(\mathsf{conf}_1,\ldots,\mathsf{conf}_{\ell+1})$ is also an $\ell$-transition in $\langle G,k,\mathsf{init},\mathsf{fin}\rangle$, so it is a \yes-instance.

Now, assume that there exists $1\leq i\leq \ell +1$ such that $v_e$ appears in $\mathsf{conf}_i$. Because $H$ does not have any vertices from $\mathsf{init}$ or $\mathsf{fin}$, we get that $v_e$ does not appear in $\mathsf{init}$ or $\mathsf{fin}$. For each $1\leq i\leq \ell+1$, denote $\mathsf{conf}_i=(v^i_1,\ldots,v^i_k)$. Then, for each $1\leq i\leq \ell +1$ and $1\leq j\leq k$ such that $v^i_j=v_e$, it follows that $v^{i-j+1}_1=v_e$. So, to cover all configurations in the given $\ell$-transition $(\mathsf{conf}_1,\ldots,\mathsf{conf}_{\ell+1})$ that contain $v_e$, it is enough to consider only configurations $\mathsf{conf}_i=(v^i_1,\ldots,v^i_k)$ such that $v^i_1=v_e$ and the $k-1$ configurations that come after them.

Now, let $1\leq i\leq \ell +1$ such that $\mathsf{conf}_i=(v^i_1,\ldots,v^i_k)$ and $v^i_1=v_e$ . By definition, $v^{i-1}_1\in N_{G/e}(v_e)$. Recall that $e=\{ u,v\}$ and $G/e$ is the graph obtained from $G$ by contracting $\{u,v\}$. Therefore, $N_{G/e}(v_e)=N_G(u)\cup N_G(v)$. Assume without loss of generality that $v^{i-1}_1\in N_G(u)$. Now, if $v^{i+1}_1\in N_G(u)$, then for every $0\leq s\leq k-1$, we can replace $v^{i+s}_{1+s}$ by $u$ and thereby define valid configurations in $\langle G,k,\mathsf{init},\mathsf{fin}\rangle$. Notice that we derive a valid transition in $\langle G,k,\mathsf{init},\mathsf{fin}\rangle$. Now, assume $v^{i+1}_1\in N_G(v)$. Then we can replace $\mathsf{conf}_i,\ldots,\mathsf{conf}_{i+k-1}$ by $\mathsf{conf}'_i,\ldots,\mathsf{conf}'_{i+k}$, where $\mathsf{conf}'_i=(u,v^i_2,\ldots,v^i_k)$, $\mathsf{conf}'_{i+1}=(v,u,v^i_2,\ldots,v)$, $\mathsf{conf}'_{i+2}=(v^{i+1}_1,v,u,v^i_2,\ldots,v^i_{k-2})$,\ldots, $\mathsf{conf}'_{i+k}=(v^{i+k-1}_1,\ldots,v^{i+k-1}_{k-1},v)$. Notice that, again, we derive a valid transition in $\langle G,k,\mathsf{init},\mathsf{fin}\rangle$.

Overall, notice that by modifying the given $\ell$-transition in this manner, we get (for some $t\geq \ell$) a $t$-transition $(\mathsf{conf}'_1,\ldots,\mathsf{conf}'_{t+1})$ in $\langle G,k,\mathsf{init},\mathsf{fin}\rangle$, such that $\mathsf{conf}'_1=\mathsf{init}$ and $\mathsf{conf}'_{t+1}=\mathsf{fin}$. Therefore, $\langle G,k,\mathsf{init},\mathsf{fin}\rangle$ is a \yes-instance.                         
\end{proof}

By combining the last two lemmas, we immediately derive the following result.

\begin{lemma} \label{lem:InsEq}
Let $\langle G,k,\mathsf{init},\mathsf{fin}\rangle$ be an instance of {\sc Snake Game}, where $G$ contains a $3\sqrt{k}$-wall $H$ having vertices from neither $\mathsf{init}$ nor $\mathsf{fin}$. Let $e=\{u,v\} \in E(H)$. Then, $\langle G,k,\mathsf{init},\mathsf{fin}\rangle$ is a \yes-instance if and only if $\langle G/e,k,\mathsf{init},\mathsf{fin}\rangle$ is a \yes-instance.   
\end{lemma}

In the last three lemmas, we assumed that the $3\sqrt{k}$-wall $H$ contains vertices from neither $\mathsf{init}$ nor $\mathsf{fin}$. In order to ensure that we have such a $3\sqrt{k}$-wall, we need to search for an $r$-wall, for a larger $r$, such that whenever $G$ contains an $r$-wall, it also contains such $H$. From the next observation, a $7k$-wall is large enough.

\begin{figure}
\centering
    \includegraphics[page=2, width=0.75\textwidth]{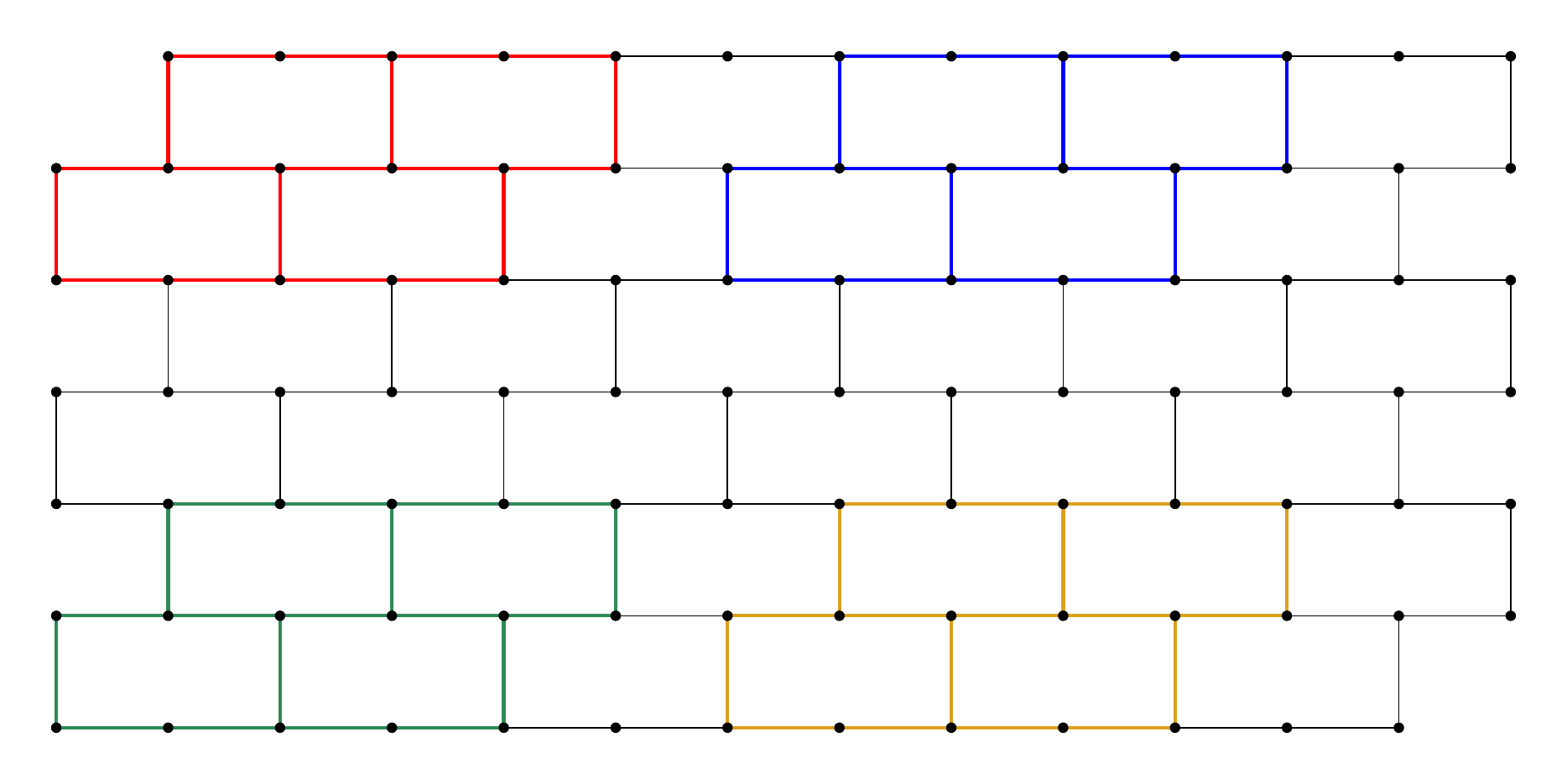}
    \caption{Four distinct elementary 3-walls in an elementary 7-wall.}
  \label{fig:7-wall}
\end{figure}   

\begin{observation} \label{obs:wallCon9}
Let $\langle G,k,\mathsf{init},\mathsf{fin}\rangle$ be an instance of {\sc Snake Game}. If $G$ contains a $7k$-wall, then it contains a $3\sqrt{k}$-wall $H$ such that $H$ does not have any vertices from $\mathsf{init}$ or $\mathsf{fin}$. Moreover, given a $7k$-wall, such a $3\sqrt{k}$-wall $H$ can be found in time polynomial in $k$.   
\end{observation}

\begin{proof}
Let $H'$ be a $7k$-wall in $G$. Notice that $H'$ contains more than $2k$ $3\sqrt{k}$-walls that have no vertices in common (see Figure \ref{fig:7-wall}), and $2k+1$ such $3\sqrt{k}$-walls can be found in polynomial time. Since $|V(\mathsf{init})\cup V(\mathsf{fin})|\leq 2k$, by the pigeonhole principle, we get that $H'$, and so does $G$, contains a $3\sqrt{k}$-wall $H$ such that $H$ does not have any vertices from $\mathsf{init}$ or $\mathsf{fin}$. In order to find such a $3\sqrt{k}$-wall, we simply examine iterate over our $2k+1$ different $3\sqrt{k}$-walls until the desired one is found. Therefore, a $3\sqrt{k}$-wall $H$ such that $H$ does not have any vertices from $\mathsf{init}$ or $\mathsf{fin}$ exists and can be found in time polynomial in $k$.          
\end{proof} 

\subsection{The Algorithm for Reduction of Treewidth}

Now, we are ready to present our algorithm for reduction of treewidth, which we call $\mathsf{TrRedAlg}$.
The input of $\mathsf{TrRedAlg}$ is an instance $\langle G,k,\mathsf{init},\mathsf{fin}\rangle$ of {\sc Snake Game} and its output should be an equivalent instance $\langle G',k,\mathsf{init},\mathsf{fin}\rangle$ where the treewidth of $G'$ is $k^{\OO(1)}$.  
Towards that, $\mathsf{TrRedAlg}$ performs at most $|V(G)|$ iterations that can each be executed in time polynomial in $|V(G)|$.
At each iteration, the algorithm either returns an instance $\langle G',k,\mathsf{init},\mathsf{fin}\rangle$ equivalent to $\langle G,k,\mathsf{init},\mathsf{fin}\rangle$ such that $|V(G')|=|V(G)|-1$, or reports that $G$ has treewidth $k^{\OO(1)}$. 
First, we present the procedure executed for each iteration, called $\mathsf{TrRedAlgIter}$ (see Algorithm \ref{alg:TreeRedAlgIter}).
For the execution of $\mathsf{TrRedAlgIter}$, we use as ``black box'' a known algorithm (Chekuri and Chuzhoy, 2016 \cite{DBLP:journals/jacm/ChekuriC16}) that (in polynomial time) either finds an $r$-wall in a graph, or reports that its treewidth is $k^{\OO(1)}$ (see Proposition \ref{prop:wallOrTree}).

\begin{proposition} \label{prop:wallOrTree}
There exists an algorithm that, given a graph $G$ and a nonnegative integer $t$, either finds a $t$-wall $H$ in $G$, or reports that $G$ has treewidth $t^{\OO(1)}$. Moreover, that algorithm runs in time $\OO(|V(G)| + |E(G)|)$.    
\end{proposition}

\begin{algorithm}
    \SetKwInOut{Input}{Input}
    \SetKwInOut{Output}{Output}
	\medskip
    {\textbf{function} $\mathsf{TrRedAlgIter}$}$(\langle G,k,\mathsf{init},\mathsf{fin}\rangle)$\;
    	
    	Use the algorithm from Proposition \ref{prop:wallOrTree} to try to find a $7k$-wall $H$ in $G$.\;
			\If {failed to find such $H$}{Report that $G$ has treewidth $k^{\OO(1)}$ and \Return\;}
			Use Observation \ref{obs:wallCon9} to find a $3\sqrt{k}$-wall $H'$ in $H$ having vertices from neither $\mathsf{init}$ nor $\mathsf{fin}$\;
			Choose an arbitrary edge $e\in E(H')$\;
			\Return $\langle G/e,k,\mathsf{init},\mathsf{fin}\rangle$\;

    \caption{$\mathsf{TrRedAlgIter}$}
    \label{alg:TreeRedAlgIter}
\end{algorithm}
    
Now, we aim to show that $\mathsf{TrRedAlgIter}$ is correct. That is, given an instance $\langle G,k,\mathsf{init},\mathsf{fin}\rangle$ of {\sc Snake Game}, if $\mathsf{TrRedAlgIter}$ reports that $G$  has treewidth $k^{\OO(1)}$, then we need to show that this is the case; otherwise, $\mathsf{TrRedAlgIter}$ returns an instance of {\sc Snake Game}, and we need to show that this instance is equivalent to $\langle G,k,\mathsf{init},\mathsf{fin}\rangle$. 

\begin{lemma} \label{lem:TrRedAlgIterCorre}
Let $\langle G,k,\mathsf{init},\mathsf{fin}\rangle$ be an instance of {\sc Snake Game}. If $\mathsf{TrRedAlgIter}(\langle G,k,\mathsf{init},\mathsf{fin}\rangle)$ reports that $G$ has treewidth $k^{\OO(1)}$, then $G$ has treewidth $k^{\OO(1)}$. Otherwise, $\mathsf{TrRedAlgIter}$ returns an equivalent instance of $\langle G,k,\mathsf{init},\mathsf{fin}\rangle$.       
\end{lemma}

\begin{proof}
First, assume that $\mathsf{TrRedAlgIter}(\langle G,k,\mathsf{init},\mathsf{fin}\rangle)$ reports that $G$ has treewidth $k^{\OO(1)}$. Then, by the definition of $\mathsf{TrRedAlgIter}$, the algorithm from Proposition \ref{prop:wallOrTree} reported that $G$ has treewidth $k^{\OO(1)}$. Therefore, by the correctness of Proposition \ref{prop:wallOrTree}, it follows that $G$ indeed has treewidth $k^{\OO(1)}$.

Now, observe that if $\mathsf{TrRedAlgIter}$ does not report that $G$ has treewidth $k^{\OO(1)}$, then it returns the instance $\langle G/e,k,\mathsf{init},\mathsf{fin}\rangle$. Then, by the definition of $\mathsf{TrRedAlgIter}$ and the correctness of Proposition \ref{prop:wallOrTree} and Observation \ref{obs:wallCon9}, $H'$ is indeed a $3\sqrt{k}$-wall in $G$ having vertices from neither $\mathsf{init}$ nor $\mathsf{fin}$. Therefore, by Lemma \ref{lem:InsEq}, $\langle G,k,\mathsf{init},\mathsf{fin}\rangle$ is a \yes-instance if and only if $\langle G/e,k,\mathsf{init},\mathsf{fin}\rangle$ is a \yes-instance. Thus, $\langle G/e,k,\mathsf{init},\mathsf{fin}\rangle)$ is equivalent to $\langle G,k,\mathsf{init},\mathsf{fin}\rangle$.           
\end{proof}  

Now, we argue that $\mathsf{TrRedAlgIter}$ run in time polynomial in $|V(G)|$.

\begin{observation} \label{lem:TrRedAlgIterTime}
The algorithm $\mathsf{TrRedAlgIter}$ runs in time $\OO(|V(G)|^2)$. 
\end{observation} 

\begin{proof}
The algorithm from Proposition \ref{prop:wallOrTree} runs in time $\OO(|V(G)| + |E(G)|)$. Given a $7k$-wall, by Observation \ref{obs:wallCon9}, a $3\sqrt{k}$-wall $H'$ having vertices from neither $\mathsf{init}$ nor $\mathsf{fin}$ can be found in time polynomial in $k\leq|V(G)|$. The other steps of $\mathsf{TrRedAlgIter}$ can be implemented in $\OO (1)$ time. As $|E(G)| = \OO(|V(G)|^2)$, the observation follows. 
\end{proof}

Now, we present our reduction algorithm $\mathsf{TrRedAlg}$ (see Algorithm \ref{alg:TreeRedAlg}).
We denote the outputs of the algorithms $\mathsf{TrRedAlgIter}$ and $\mathsf{TrRedAlg}$ for an instance $I$ by $\mathsf{TrRedAlgIter}(I)$ and $\mathsf{TrRedAlg}(I)$ respectively.  

\begin{algorithm}[t]
    \SetKwInOut{Input}{Input}
    \SetKwInOut{Output}{Output}
	\medskip
    {\textbf{function} $\mathsf{TrRedAlgIter}$}$(\langle G,k,\mathsf{init},\mathsf{fin}\rangle)$\;
    	
    	 $G^I\gets G$\;
			\While{$\mathsf{TrRedAlgIter}(I=\langle G^I,k,\mathsf{init},\mathsf{fin}\rangle)$ does not report $G^I$ has treewidth $k^{\OO(1)}$}{$I\gets \mathsf{TrRedAlgIter}(I)$\;}
			\Return $I$ \;

    \caption{$\mathsf{TrRedAlg}$}
    \label{alg:TreeRedAlg}
\end{algorithm}

Now, we show that $\mathsf{TrRedAlg}$ runs in time polynomial in $|V(G)|$. 

\begin{lemma} \label{lem:TrRedAlgTime}
The algorithm $\mathsf{TrRedAlg}$ runs in time $\OO(|V(G)|^3)$. 
\end{lemma}

\begin{proof}
Consider an input instance $\langle G,k,\mathsf{init},\mathsf{fin}\rangle$. At each iteration $t\geq 2$, $\mathsf{TrRedAlg}$ calls $\mathsf{TrRedAlgIter}$ with an instance $\langle G^t,k,\mathsf{init},\mathsf{fin}\rangle$ such that $|V(G^t)|= |V(G^{t-1})|-1$. By Lemma \ref{lem:TrRedAlgIterTime}, $\mathsf{TrRedAlgIter}$ works in time polynomial in $|V(G^t)|$. Therefore, as $G^1=G$, there are at most $|V(G)|$ iterations. Thus, we conclude that $\mathsf{TrRedAlg}$ runs in time $\OO(|V(G)|^3)$.    
\end{proof}

Now, we show that for a given instance of {\sc Snake Game} $\langle G,k,\mathsf{init},\mathsf{fin}\rangle$, $\mathsf{TrRedAlg}$ $(\langle G,k,\mathsf{init},\mathsf{fin}\rangle) = \langle G',k,\mathsf{init},\mathsf{fin}\rangle$ is an equivalent to $\langle G,k,\mathsf{init},\mathsf{fin}\rangle$, and also $G'$ has treewidth $k^{\OO(1)}$.   

\begin{lemma} \label{lem:TrRedAlgCorre}
Let $\langle G,k,\mathsf{init},\mathsf{fin}\rangle$ be an instance of {\sc Snake Game}. Then $\mathsf{TrRedAlg}(\langle G,k,\mathsf{init},\mathsf{fin}\rangle)=\langle G',k,\mathsf{init},\mathsf{fin}\rangle$ is an equivalent instance of $\langle G,k,\mathsf{init},\mathsf{fin}\rangle$, and $G'$ has treewidth $k^{\OO(1)}$.
\end{lemma}  

\begin{proof}
By Lemma \ref{lem:TrRedAlgIterCorre}, at each iteration of $\mathsf{TrRedAlg}$, $\mathsf{TrRedAlgIter}$ returns an instance of {\sc Snake Game}, that is equivalent to the previous one. Therefore, as the first instance is $\langle G,k,\mathsf{init},\mathsf{fin}\rangle$, the output of $\mathsf{TrRedAlg}$ denoted by $\langle G',k,\mathsf{init},\mathsf{fin}\rangle$, is equivalent to $\langle G,k,\mathsf{init},\mathsf{fin}\rangle$. By the definition of $\mathsf{TrRedAlg}$ and Lemma \ref{lem:TrRedAlgIterCorre}, it directly follows that $G'$ has treewidth $k^{\OO(1)}$.               
\end{proof}

Finally, we combine the last two lemmas to conclude the correctness of our algorithm of the reduction.

\begin{theorem}\label{the:treewidthRed}
Given any instance $\langle G,k,\mathsf{init},\mathsf{fin}\rangle$ of {\sc Snake Game}, $\mathsf{TrRedAlg}$ returns an equivalent instance $\langle G',k,\mathsf{init},\mathsf{fin}\rangle$ such that $G'$ has treewidth $k^{\OO(1)}$. Moreover, $\mathsf{TrRedAlg}$ runs in time $\OO(|V(G)|^3)$.   
\end{theorem} 

Observe that, a graph $G$ on $n$ vertices and with treewidth $\tw$ has $\OO(\tw^2 n)$ edges as the tree decomposition of $G$ can have at most $n$ bags and every bag contributes at most $\tw^2$ edges. Thus, the following corollary follows directly from the Theorem~\ref{the:treewidthRed}.

\begin{corollary}\label{cor:edgeSetSize}
Given any instance $\langle G,k,\mathsf{init},\mathsf{fin}\rangle$ of {\sc Snake Game}, $\mathsf{TrRedAlg}$ returns an equivalent instance $\langle G',k,\mathsf{init},\mathsf{fin}\rangle$ such that $|E(G')| = k^{\OO(1)}|V(G)'|$.
\end{corollary}

\section{Conclusion and Directions for Future Research}\label{sec:conclusion}

In this paper, we presented a coherent picture of the parameterized complexity of the {\sc Snake Game} problem. As our main contribution, we proved that {\sc Snake Game} is \FPT. This result was complemented by the proofs that {\sc Snake Game} is unlikely to admit a polynomial kernel even on grid graphs, but admits a treewidth-reduction procedure. For our main contribution, we presented a novel application of the method of color-coding by utilizing it to sparsify the configuration graph of the problem. To some extent, our paper can be considered as pioneering work in the study of the parameterized complexity of motion planning problems where intermediate configurations are of importance (which has so far been largely neglected), and may lay the foundations for further research of this topic. In this regard, the choice of {\sc Snake Game} was a natural starting point given that it involves only one mobile object as well as captures only the most basic property of a physical object---it cannot intersect itself and other obstacles.

We conclude the paper with a few directions for further research. As a first step to broaden the scope of this kind of research, we suggest to consider the parameterized complexity of motion planning problems that involve several mobile agents. Possibly, not all mobile agents can be controlled---that is, the patterns of motion of some of them may be predefined. The controllable mobile agents themselves can have either a joint objective or independent tasks. Arguably, the most basic requirement of the mobile agents is to avoid collision with one another while they perform these tasks. We remark that a natural parameter in this regard is the number of mobile agents. Second, apart (or in addition to) the study of multiple agents, we also find it interesting to consider other types of motion as well as other shapes of robots. Moreover, one can study objectives that are more complicated than merely reaching one position from another, particularly objectives where the mobile agent can effect the environment rather than only move within it. Lastly, we also suggest to consider the study of the parameterized complexity of motion planning problems in geometric settings in 3D---while graphs provide a useful abstraction, some problems may require a more accurate representation of the robots and physical environment.
\bibliographystyle{siam}
\bibliography{Refs}

\end{document}